\documentclass[11pt]{article}

\usepackage[T1]{fontenc}

\usepackage[margin=1in]{geometry}

\usepackage[numbers,sort]{natbib}

\usepackage{pstricks}
\usepackage{subfigure}

\usepackage{times}
\usepackage{color}
\usepackage[ruled,vlined,linesnumbered]{algorithm2e}
\usepackage{pstricks, pst-node, pst-tree}
\usepackage{subfigure}
\usepackage{amssymb,amsmath,stmaryrd}

\newtheorem{theorem}{Theorem}
\newtheorem{definition}[theorem]{Definition}
\newtheorem{lemma}[theorem]{Lemma}
\newtheorem{claim}[theorem]{Claim}
\newtheorem{fact}[theorem]{Fact}
\newtheorem{example}[theorem]{Example}

\newcommand{\sq}{\hbox{\rlap{$\sqcap$}$\sqcup$}}
\newcommand{\qed}{\hspace*{\fill}\sq}
\newenvironment{proof}{\noindent {\bf Proof.}\ }{\qed\par\vskip 4mm\par}

\newcommand{\seq}[1]{#1}

\newcommand{\tuple}[1]{\langle#1\rangle}

\def\rmIndex{\mathsf{rmIdx}}
\def\brmIndex{\mathsf{\widehat{rm}Idx}}
\def\rmLastIndex{\mathsf{rmLastIdx}}
\def\crpIndex{\mathsf{crLastIdx}}
\def\crpFirstIndex{\mathsf{crFirstIdx}}
\def\lmIndex{\mathsf{lmIdx}}
\def\leftIndex{\mathsf{leftIdx}}
\def\rightIndex{\mathsf{rightIdx}}
\def\l{\mathsf{left}}
\def\r{\mathsf{right}}
\def\rm{\mathsf{rm}}
\def\lm{\mathsf{lm}}
\def\crp{\mathsf{cr}}
\def\brm{\mathsf{\widehat{rm}}}
\def\occ{\mathit{occ}}
\def\nodes{\mathit{nodes}}
\def\len{\mathit{len}}

\def\cost{w}
\def\opt{\mathit{opt}}
\def\feasible{\mathit{feasible}}
\def\R{\mathcal{R}}
\def\C{\mathcal{C}}

\def\GETS{\leftarrow}
\def\cw{\mathsf{cw}}
\def\acw{\mathsf{acw}}
\def\cwIndex{\mathsf{cwIdx}}
\def\acwIndex{\mathsf{acwIdx}}
\def\sym{\mathsf{sym}}

\def\min{\mathsf{min}}
\def\max{\mathsf{max}}

\title{Ride Sharing with a Vehicle of Unlimited Capacity\\ (full version)}

\author{
  Angelo Fanelli\footnote{
  CNRS, France.  Email: \texttt{angelo.fanelli@unicaen.fr}
  }
  \and
  Gianluigi Greco\footnote{
  Department of Mathematics and Computer Science, University of Calabria, Italy.
  Email:  \texttt{ggreco@mat.unical.it}
  }
}
\date{}

\begin{document}

\maketitle

\begin{abstract}
A ride sharing problem is considered where we are given a graph, whose edges are equipped with a travel cost, plus a set of objects, each
associated with a transportation request given by a pair of origin and destination nodes. A vehicle travels through the graph, carrying each
object from its origin to its destination without any bound on the number of objects that can be simultaneously transported. The vehicle starts
and terminates its ride at given nodes, and the goal is to compute a minimum-cost ride satisfying all requests. This ride sharing problem is
shown to be tractable on paths by designing a $O(h \log h+n)$ algorithm, with $h$ being the number of distinct requests and with $n$ being the
number of nodes in the path. The algorithm is then used as a subroutine to efficiently solve instances defined over cycles, hence covering all
graphs with maximum degree $2$. This traces the frontier of tractability, since $\bf NP$-hard instances are exhibited over trees whose maximum
degree is $3$.
\end{abstract}

\section{Introduction}\label{sec:intro}
{Vehicle routing problems} have been drawn to the attention of the research community in the late 50's~\cite{doi:10.1287/mnsc.6.1.80}. Since
then, they have attracted much attention in the literature due to their pervasive presence in real-world application scenarios, till becoming
nowadays one of the most studied topics in the field of operation research and combinatorial optimization (see,
e.g.,~\cite{Laporte1992345,doi:10.1137/1.9780898718515,Eksioglu:2009:SVR:1651928.1652074} and the references therein).

Within the broad family of vehicle routing problems, a noticeable class is constituted by the {pickup and delivery problems}, where a given set
of objects, such as passengers or goods, have to be {picked} at certain nodes of a transportation network and {delivered} at certain
destinations~\cite{Cordeau_transportationon}. Pickup and delivery problems can be divided in two main groups~\cite{PDP}.
The first group refers to situations where we have a single type of object to be transported, so that pickup and delivery locations are
unpaired (see, e.g., \cite{OCPD}).
The second group deals, instead, with problems where each transportation request is associated with a specific origin and a specific
destination, hence resulting in paired pickup and delivery points (see, e.g.,~\cite{Kalantari1985377,Dumas19917}).

In the paper, we focus on problems of the latter kind, and we deal with the most basic setting where \emph{one vehicle} is available only. The
vehicle is initially located at some given source node and it must reach a given destination node by means of a \emph{feasible} ride, that is,
of a ride satisfying all requests. The edges of the network are equipped with weights, and the goal is to compute an \emph{optimal} ride, that
is, a feasible ride minimizing the sum of the weights of the edges traversed by the vehicle.

Ride sharing with one vehicle has attracted much research in the literature and most of the foundational results in the area of vehicle routing
precisely refer to this setting---see Section~\ref{sec:related}.
In fact, earlier works have mainly focused on the case where the capacity of the vehicle is bounded by some given constant. But, there are
application scenarios where the capacity of the vehicle can be better thought as being \emph{unlimited}, as it happens, for instance, when we
are transporting intangible objects, such as messages. More generally, we might know beforehand that the number of objects to be transported is
less than the capacity of the vehicle; and, accordingly, we would like to use solution algorithms that are more efficient than those proposed
in the literature and designed in a way that this knowledge is not suitably taken into account.

The goal of the paper is to fill this gap, and to study complexity and algorithmic issues arising with ride sharing problems in presence of one
vehicle of unlimited capacity. The analysis has been conducted by considering different kinds of \emph{undirected} graph topologies, which have
been classified on the basis of the degree of their nodes. Let $n$ be the number of nodes in the underlying graph, let $q$ be the number of
requests (hence, of objects to be transported), and let $h$ denote the number of distinct requests (so, $h\leq q$ and $h\leq n^2$). Then, our
results can be summarized as follows:
\begin{itemize}
  \item[$\blacktriangleright$] Optimal rides can be computed in polynomial time over graphs that are \emph{paths}. In particular, an
      algorithm is exhibited to compute an optimal ride in $O(h  \log h+n)$. This improves the $O(q n+n^2)$ bound that we obtain with the
      state-of-the-art algorithm by Guan and Zhu~\cite{guan1998multiple} for vehicles with limited capacity, by na\"ively setting the limit
      to $k$.

  \item[$\blacktriangleright$] The design and the analysis of the above algorithm is the main technical achievement of the paper. By using
      the algorithm as a basic subroutine, we are then able to show that optimal rides can be computed in polynomial time over
      \emph{cycles} too, formally in $O(m^2 \cdot (h  \log h+n))$, with $m$ being the number of distinct nodes that are endpoints of some
      request, so that $m\leq 2h$ and $m\leq n$. The result has no counterpart in the limited capacity setting, where no polynomial time
      algorithm over cycles has been exhibited so far---special cases have been actually addressed, as discussed in
      Section~\ref{sec:related}.

  \item[$\blacktriangleright$] Path and cycles completely cover all graphs whose maximum degree is 2. In fact, this value precisely traces
      the frontier of tractability for the ride sharing problem we have considered, as $\bf NP$-hard instances are exhibited over graphs
      whose maximum degree is 3 and which are moreover trees.
\end{itemize}

\noindent The rest of the paper is organized as follows. The formal framework and some basic results are illustrated in
Section~\ref{sec:prelim}. The algorithms for paths and cycles are presented and their complexity is analyzed in Section~\ref{sec:path} and
Section~\ref{sec:cycles}, respectively. A discussion of relevant related works is reported in Section~\ref{sec:related}, while a few concluding
remarks are discussed in Section~\ref{sec:conclusion}.

\section{Ride Sharing Scenarios}\label{sec:prelim}
\subsection{Formal Framework}

Let $G=(V, E, w)$ be an undirected weighted graph, where $V$ is a set of {nodes} and $E$ is a set of {edges}. Each edge $e\in E$ is a set
$e\subseteq V$ with $|e|=2$, and it is equipped with a cost $w(e)\in \mathbb{Q}^+$.
%

A \emph{ride} $\pi$ in $G$ is a sequence of nodes $\seq{\pi_1, \ldots, \pi_k}$ such that $\pi_i\in V$ is the node reached at the \emph{time
step} $i$ and $\{\pi_i, \pi_{i+1}\} \in E$, for each $i$ with $1\leq i\leq k-1$. The time step $k> 0$ is called the \emph{length} of $\pi$,
hereinafter denoted by $\len(\pi)$.
The value $\sum_{i=1}^{k-1} w(\{\pi_i, \pi_{i+1}\})$ is the \emph{cost} of $\pi$ (w.r.t.~$w$) and is denoted by $\cost(\pi)$.
Moreover, $\nodes(\pi)$ denotes the set of all nodes $v\in V$ occurring in $\pi$.

A \emph{request} on $G=(V, E, w)$ is a pair $(s,t)$ such that $\{s,t\}\subseteq V$. Note that $s$ and $t$ are not necessarily distinct, and
they are called the starting and terminating nodes, respectively, of the request.
We say that a ride $\pi$ in $G$ \emph{satisfies} the request $(s,t)$ if there are two time steps $i$ and $i'$ such that $1\leq i\leq i'\leq
\len(\pi)$, $\pi_i=s$ and $\pi_{i'}=t$. If $\mathcal{C}$ is a set of requests on $G$, then $V_\mathcal{C}$ is the set of all starting and
terminating nodes occurring in it.

A \emph{ride-sharing} scenario consists of a tuple $\mathcal{R}=\tuple{G, (s_0,t_0),\mathcal{C}}$, where $G=(V, E, w)$ is an undirected
weighted graph, $(s_0,t_0)$ is a request on $G$ and $\mathcal{C}$ is a non-empty set of requests.
%
A ride $\pi=\seq{\pi_1, \ldots, \pi_k}$ in $G$ is \emph{feasible} for $\mathcal{R}$ if $\pi_1=s_0$, $\pi_k=t_0$, and $\pi$ satisfies each
request in $\mathcal{C}$.
The set of all feasible rides for $\mathcal{R}$ is denoted by $\feasible(\mathcal{R})$.
A feasible ride $\pi$ is \emph{optimal} if $\cost(\pi')\geq \cost(\pi)$, for each feasible ride $\pi'$. The set of all optimal rides for
$\mathcal{R}$ is denoted by $\opt(\mathcal{R})$.

Let $\mathcal{R}=\tuple{G, (s_0,t_0),\mathcal{C}}$ be a ride-sharing scenario, and let $\pi$ be a ride in $G$. Let $i$ and $i'$ be two time
steps such that $1\leq i\leq i'\leq \len(\pi)$. Then, we denote by $\pi[i,i']$ the ride $\pi_i,\ldots,\pi_{i'}$ obtained as the sequence of the
nodes occurring in $\pi$ from time step $i$ to time step $i'$.
%
If $\pi$ and $\pi'$ are two rides on $G$, then we write $\pi'\preceq \pi$ if either $\pi'=\pi$ or, recursively, if there are two time steps $i$
and $i'$ such that $1\leq i< i'\leq \len(\pi)$, $\pi_{i+1}=\pi_{i'}$ or $\pi_{i}=\pi_{i'-1}$, and $\pi' \preceq \pi[1,i],\pi[i',\len(\pi)]$ (informally speaking, $\pi'$ can be obtained from $\pi$ by removing a subsequence of nodes).

\begin{fact}\label{lem:preserving}
Let $\pi$ and $\pi'$ be two rides such that $\pi'\preceq \pi$. Then:  $w(\pi')\leq w(\pi)$;  if $\pi'$ satisfies a request $(s,t)\in
\mathcal{C}$, then $\pi$ satisfies $(s,t)$, too;  if $\pi$ is feasible (resp., optimal) and $V_\mathcal{C}\cap (\nodes(\pi)\setminus
\nodes(\pi')) =\emptyset$, then $\pi'$ is feasible (resp., optimal), too.
\end{fact}

\begin{example}\label{example:prelim}\em
Consider the following instance (depicted in Figure \ref{fig5}): $V=\{1, 2, 3, 4, 5, 6 \}$, $E=\{\{1,2\}, \{1,4\},$ $\{2,3\}, \{2,5\}, \{3,4\}, \{3,6\}, \{4,5\}, \{5,6\} \}$, $w(e) = 1\,$ for every $e\in E$, $s_0 =1$, $t_0=2$, and $\mathcal{C}=\{(1,5), (6,2) \}$.

The ride $\pi_1 = 1,4,5,2$ is not feasible because it does not satisfy the request $(6,2)$. Instead, $\pi_2 = 1,2,3,4,5,6,5,4,3,2$ is feasible and its cost is $9$.
Nevertheless, $\pi_2$ this is not an optimal ride, because $\pi_3 = 1,4,5,6,3,2$  (thick red edges in Figure \ref{fig5}) is also feasible and its cost is $5$; in particular, note that $\pi_3 \preceq \pi_2$ and that $\pi_3$ is an optimal ride. \hfill $\lhd$
\end{example}

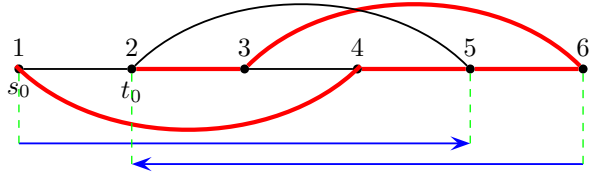
\begin{figure}[t]
\centering
\def\H{3}
\def\L{1.2}
\def\M{0.7}

\def\a{1}
\def\b{3}
\def\c{5}
\def\d{7}
\def\e{9}
\def\f{11}

\def\size{\small}

\psset{xunit=.75cm}
\psset{yunit=.55cm}

\psset{dash=3pt 3pt}
\psset{arrowsize=4pt 2}

\begin{center}
\begin{pspicture}(\a,\L)(\f,3)

\psset{linewidth=.7pt}
\psset{linecolor=black}

\dotnode(\a,\H){1}\uput[90](\a,\H){\size $1$}\uput[-90](\a,\H){\size $s_0$}
\dotnode(\b,\H){2}\uput[90](\b,\H){\size $2$}\uput[-90](\b,\H){\size $t_0$}
\dotnode(\c,\H){3}\uput[90](\c,\H){\size $3$}
\dotnode(\d,\H){4}\uput[90](\d,\H){\size $4$}
\dotnode(\e,\H){5}\uput[90](\e,\H){\size $5$}
\dotnode(\f,\H){6}\uput[90](\f,\H){\size $6$}

\ncLine{-}{1}{2}
\ncLine[linecolor=red,linewidth=1.6pt]{-}{2}{3}
\ncLine{-}{3}{4}
\ncLine[linecolor=red,linewidth=1.6pt]{-}{4}{5}
\ncLine[linecolor=red,linewidth=1.6pt]{-}{5}{6}

\psset{ncurv=.7, angleB=135, angleA=45}
\nccurve[linecolor=red,linewidth=1.6pt]{3}{6}
\nccurve{2}{5}
\psset{ncurv=-.7, angleA=135, angleB=45}
\nccurve[linecolor=red,linewidth=1.6pt]{1}{4}

\pnode(\a,\L){s1}
\pnode(\e,\L){t1}
\pnode(\f,\M){s2}
\pnode(\b,\M){t2}

\psset{linecolor=blue}
\ncLine{->}{s1}{t1}
\ncLine{->}{s2}{t2}

\psset{linewidth=.5pt}
\psset{linecolor=green}
\psset{linestyle=dashed}

\ncLine{-}{1}{s1}
\ncLine{-}{5}{t1}
\ncLine{-}{6}{s2}
\ncLine{-}{2}{t2}

\end{pspicture}
\end{center}
\caption{Instance of Example \ref{example:prelim}.}
\label{fig5}
\end{figure}

\subsection{Basic Complexity Results}

It is easily seen that computing optimal rides is an intractable problem ($\bf NP$-hard), for instance, by exhibiting a reduction from the
well-known traveling salesman problem (see, e.g., \cite{Garey:1979:CIG:578533}). We start our elaboration by strengthening this result and by
showing that intractability still holds over ride-sharing scenarios defined over \emph{trees} whose maximum degree is $3$.

%

%
%

\begin{theorem}\label{thm:NP}
Computing an optimal ride is {\em $\bf NP$}-hard on scenarios $\tuple{G,(s_0,t_0),\mathcal{C}}$ such that $G$ is a tree whose maximum
degree is 3. 
\end{theorem}
\begin{proof}
Consider the following well-known $\bf NP$-hard problem: We are given a directed and connected graph $\hat G=(\hat V,\hat E)$ and a natural
number $k>0$. We have to decide whether there is a \emph{feedback vertex set} $S\subseteq \hat V$ of at most $k$ vertices, i.e., such that
$|S|\leq k$ and the graph $\hat G_S=(\hat V\setminus S, \{ (u,v)\in \hat E \mid \{u,v\}\subseteq \hat V\setminus S\})$ is acyclic. W.l.o.g.,
assume that there is a natural number $n$ such that $|\hat V|=2^n$ and that each vertex has at least one outgoing edge.

Based on $\hat G$, we adapt a reduction that can be found in~\cite{Guan199841} in order to build a ride sharing scenario $\hat
\R=\tuple{G,(s_0,t_0),\mathcal{C}}$, with $G=(V,E,w)$, as follows. First, $G$ is a binary tree rooted at a node $\hat s$ and whose leafs are
the vertices in $\hat V$; so, we have $V\supseteq \hat V$. Second,
the starting and terminating activity coincide with the root, i.e., $s_0=t_0=\hat s$. Third, for each edge $(u,v)\in \hat E$, the request
$(u,v)$ is in $\mathcal{C}$; and, no further request is in $\mathcal{C}$. Finally, $w$ is the function mapping each edge to 0, but the edges
incident to the leafs whose associated cost is $1$.
We now claim that: \emph{there is a feedback vertex set $S$ with $|S|\leq k$ $\Leftrightarrow$ there is a feasible ride $\pi$ with $w(\pi)\leq
2\times (k+|\hat V|)$}.

\begin{description}
\item[($\Rightarrow$)] Assume that $S$ is a feedback vertex set with $|S|= h\leq k$.  Consider the ride $\pi$ defined as follows. For each
    node $v\in \hat V$, let $\pi[v]$ be the ride starting at $\hat s$ reaching $v$ and going back to $\hat s$ along the unique path
    connecting them in $G$. Then, let $\pi$ be any ride having the form $\pi[\alpha_1],\dots,\pi[\alpha_{h+|\hat V|}]$ where:
    $\{\alpha_1,..,\alpha_h\}=\{\alpha_{|\hat V|+1},...,\alpha_{|\hat V|+h}\}=S$, $\{\alpha_{h+1},...,\alpha_{|\hat V|}\}=\hat V\setminus
    S$, and $\alpha_{h+1}$,...,$\alpha_{|\hat V|}$ is any topological ordering of the acyclic graph $\hat G_S$. Note that $w(\pi)=2\times
    (h+|\hat V|)\leq 2\times (k+|\hat V|)$. Moreover, $\pi$ is feasible. Indeed, consider the request $(u,v)\in \mathcal{C}$, associated
    with the edge $(u,v)\in \hat E$. We claim that there are two indices $i$ and $j$ such that $i<j$, $\alpha_i=u$, and $\alpha_j=v$, so
     that the request is satisfied by $\pi$. Indeed, if $u\in \hat V\setminus S$ and $v\in S$, then two indices enjoying these properties
     exist with $h<i\leq |\hat V|$ and $|\hat V|<j$. If $u\in \hat V\setminus S$ and $v\in \hat V\setminus S$, then  $(u,v)$ is also an
     edge in $\hat G_S$ and, by definition of topological ordering, two indices enjoying these properties exist with $h<i<j\leq |\hat V|$.
     Finally, if $u\in S$, then the desired indices are such that $i\leq h$ and $j>h$.

\item[($\Leftarrow$)] Assume that $\pi$ is a feasible ride with $w(\pi)\leq 2\times (k+|\hat V|)$. Since $\hat G$ is connected and each
    vertex has at least one outgoing edge, for each vertex $u\in \hat V$, a request of the form $(u,v)$ is in $\mathcal{C}$. Therefore, the
    edge in $G$ incident to $u$ must be traversed at least twice by $\pi$, because $G$ is a tree rooted at $\hat s=s_0=t_0$ and $u$ is a
    leaf. Therefore, we get $w(\pi)\geq 2\times |\hat V|$. Now, consider any set $\{v_1,...,v_h\}$ inducing a cycle over $\hat G$. In order
    to satisfy the requests associated with them, it must be the case that at least one vertex from this cycle, say $v_1$, occurs in two
    non-adjacent time steps of $\pi$. Hence, the edge in $G$ incident to $v_1$ is traversed at least 4 times. Given that $w(\pi)\leq
    2\times (k+|\hat V|)$, we then conclude that there is a set $S$ of $k$ vertices that cover all the cycles of the graph. This set if a
    feedback vertex set.
\end{description}

Given the properties above, the result is established as the reduction is feasible in polynomial time.
\end{proof}

Motivated by the above bad news, the rest of the paper is devoted to analyze ride-sharing scenarios over graphs whose maximum degree is 2. In
fact, these graphs must be either paths or cycles.\footnote{The case of maximum degree equals to 1 is trivial.}

\section{Optimal Rides on Paths}\label{sec:path}
In this section we describe an algorithm that, given as input a ride-sharing scenario $\mathcal{R}=\tuple{G,(s_0,t_0),\mathcal{C}}$ where
$G=(V, E, w)$ is a \emph{path}, returns an optimal ride for $\mathcal{R}$. In order to keep notation simple, we assume that nodes in $V$ are
(indexed as) natural numbers, so that $V=\{1,\ldots,n\}$. Hence, for each node $v\in V\setminus \{n\}$, the edge $\{v, v+1\}$ is in $E$; and no
further edge is in $E$. Moreover, let us define $\l(\mathcal{R})=\min_{v\in V_\mathcal{C}} v$ and $\r(\mathcal{R})=\max_{v\in V_\mathcal{C}}
v$, as the extreme (left and right) endpoints of any request in $\mathcal{C}$. 

Based on these notions, we distinguish two mutually exclusive cases:
\begin{description}
\item[``outer'':] where either $s_0\leq \l(\mathcal{R})\leq \r(\mathcal{R})\leq t_0$ or $t_0\leq \l(\mathcal{R})\leq \r(\mathcal{R})\leq
    s_0$; that is, the starting and the terminating nodes $s_0$ and $t_0$ are not properly included in the range
    $\{\l(\mathcal{R}),...,\r(\mathcal{R})\}$.

\item[``inner'':] where $\{s_0,t_0\} \cap \{ v\in V \mid \l(\mathcal{R})< v < \r(\mathcal{R})\}\neq \emptyset$; in particular, in this
    case, $\l(\mathcal{R})<\r(\mathcal{R})$ necessarily holds.
\end{description}

In the following two subsections we describe methods to address the two different cases, while their complexity will be later analyzed in
Section~\ref{sec:implementation}.
%
A basic ingredient for both methods is the concept of concatenation of rides, which is formalized below.

\begin{definition}\em
Let $\pi = \seq{\pi_1, \ldots, \pi_k}$ and $\pi' = \seq{\pi'_1, \ldots, \pi'_h}$ be two rides. Their \emph{concatenation} $\pi\mapsto \pi'$ is
the ride inductively defined as follows:
\begin{itemize}
  \item if $\pi_k=\pi_1'$ and $h>1$, then $\pi\mapsto \pi'=\seq{\pi_1, \ldots, \pi_k, \pi'_2, \ldots, \pi'_h}$;

  \item if $\pi_k=\pi_1'$ and $h=1$, then $\pi \mapsto \pi' = \pi$;

  \item if $\pi_k\neq \pi_1'$, then $\pi\mapsto \pi'$ is defined as the concatenation\footnote{When concatenating more than two sequences,
      the specific order of application of the operator $\mapsto$ is immaterial. Hence, we often avoid the use of parenthesis.} $\pi
      \mapsto \bar \pi \mapsto \pi'$, where $\bar \pi=\pi_k,\dots,\pi_1'$ is the ride obtained as the sequence of nodes connecting $\pi_k$
      and $\pi_1'$ with the smallest length.
      Note that $\bar \pi$ is univocally determined on paths.\hfill $\Box$
\end{itemize}
\end{definition}

For instance, the concatenation $1\mapsto 5\mapsto 3$ succinctly denotes the path $1,2,3,4,5,4,3$.

\subsection{Solution to the ``outer'' case}\label{sec:outer}

\IncMargin{1em}
\begin{algorithm}[t]
\SetKwInput{KwData}{Input} \SetKwInput{KwResult}{Output}

\Indm \KwData{A scenario $\mathcal{R}=\tuple{G,(s_0,t_0),\mathcal{C}}$, where $G=(V, E, w)$ is a path,\\ and with $s_0\leq
\l(\mathcal{R})\leq \r(\mathcal{R})\leq t_0$ or $t_0\leq \l(\mathcal{R})\leq \r(\mathcal{R})\leq s_0$;}%
\KwResult{An optimal ride for $\mathcal{R}$;}

\Indp

   \uIf{$s_0>t_0$}{\label{alg:step0} $\pi \GETS$ {\sc RideOnPath\_Outer}$(\sym(\R))$\;

        \Return{$\sym(\pi)$}\; }\Else{

        $\mathcal{C}^* = \{(s_1,t_1),\dots,(s_h,t_h)\} \GETS${\sc Normalize}($\mathcal{C}$)\tcc*{$s_1\leq s_2\dots \leq s_h$}\label{alg:stepbase}

	\Return{$s_0 \mapsto s_1 \mapsto t_1 \mapsto s_2 \mapsto \ldots \mapsto s_h \mapsto t_h \mapsto t_0$}\;\label{alg:lastStep}
	}
	\caption{{\sc RideOnPath\_Outer}}
	\label{alg:canonical_simple}
\end{algorithm}

Consider Algorithm~\ref{alg:canonical_simple}, named {\sc RideOnPath\_Outer}. In the first step, it distinguishes the case $s_0>t_0$ from the
case $s_0\leq t_0$. Indeed, the former can be reduced to the latter by introducing the concept of \emph{symmetric} scenario. For every node
$v\in V$, let $\sym(v)=n-v+1$. Denote by $\sym(\pi)$ and $\sym(\mathcal{C})$ the ride and the set of requests derived from the ride $\pi$ and
the set of requests $\mathcal{C}$, respectively, by replacing each node $v$ with its ``symmetric'' counterpart $\sym(v)$. Finally, denote by
$\sym(\R)$ the scenario $\tuple{G,(\sym(s_0),\sym(t_0)),\sym(\mathcal{C})}$, referred to as the symmetric scenario of $\R$. Then, the following
is immediately seen to hold.
\begin{fact}\label{fact:symmetric}
Let $\pi$ be a ride. Then, $\pi$ is an optimal ride for $\R$ if, and only if, $\sym(\pi)$ is an optimal ride for $\sym(\R)$.
\end{fact}
According to the previous observation, step~\ref{alg:stepbase} and step~\ref{alg:lastStep} are the core of the computation by addressing the
case $s_0\leq t_0$, where hence $s_0\leq \l(\mathcal{R})\leq \r(\mathcal{R})\leq t_0$. The idea is to reduce the set of requests $\mathcal{C}$
to an ``equivalent'' set of requests  $\mathcal{C}^*$, which presents a simpler structure that we call \emph{normal form}. Formally,  let
$\mathcal{C}^*=\{(s_1,t_1),\dots,(s_h,t_h)\}$, and let us say that $\mathcal{C}^*$ is in normal form if $t_i<s_i$ for each $i\in
\{1,\dots,h\}$, and $s_i<t_{i+1}$ for each $i\in\{1,\dots,h-1\}$.
The reduction is performed at step~\ref{alg:stepbase}, where {\sc Normalize} is invoked. In Lemma~\ref{claim:normalize}, we shall show that the
corresponding normal form preserves optimal solutions, i.e., every optimal solution with respect to the normal form is also an optimal solution
with respect to the original set of requests. The advantage of having a set of requests in normal form is the inherent simplicity in deriving
an optimal solution.  At step~\ref{alg:lastStep} the algorithm returns the optimal solution with respect to the normal form,  whose optimality
will be proven in Theorem \ref{thm:nomralized}.
Now, we shall take a closer and more formal look at these steps, by also illustrating their executions on a simple scenario in
Example~\ref{example:line1} and Example~\ref{example:line2}, respectively.

\smallskip

Step~\ref{alg:stepbase} in {\sc RideOnPath\_Outer} reduces the set of requests $\mathcal{C}$ to a normal form by invoking {\sc Normalize}.

 The definition of {\sc Normalize} is shown in Algorithm~\ref{alg:normalize}: Step~\ref{alg:step1} is responsible of filtering out all requests
$(s,t)$ such that $s\leq t$.
Steps~\ref{alg:step3:1} and \ref{alg:step3:2} iteratively ``merge'' all pairs of requests $(s,t)$ and $(s',t')$ such that $t<s$, $t'<s'$ and
$t' \leq t \leq s' \leq s$. Finally, steps~\ref{alg:step2:1} and \ref{alg:step2:2} remove all requests $(s,t)$ with $t<s$ and for which there
is a request $(s',t')$ such that  $t'\leq t < s \leq s'$. In the next lemma we show that the set of requests $\mathcal{C}^*$ returned by {\sc
Normalize} is in normal form and that the optimal ride for the ride-sharing scenario $\tuple{G,(s_0,t_0),\mathcal{C}^*}$ is  an optimal ride
also for $\R$.

\begin{algorithm}[t]
\SetKwInput{KwData}{Input} \SetKwInput{KwResult}{Output}

\Indm \KwData{A set $\mathcal{C}$ of requests  with $s_0\leq \l(\mathcal{R})\leq \r(\mathcal{R})\leq t_0$;}%
\KwResult{A set of requests $\mathcal{C}^*$ in normal form and such that $\opt(\tuple{G,(s_0,t_0),\mathcal{C}^*})\subseteq \opt(\R)$;}

\Indp
	$\mathcal{C}^* \GETS  \mathcal{C}  \setminus \{(s,t) \mid s\leq t\}$\;\label{alg:step1}
		
	\While{exist $(s, t), (s', t') \in {\mathcal{C}^*}$ such that $t<s$, $t'<s'$, and $t' \leq t \leq s' \leq s$}{\label{alg:step3:1}
	$ {\mathcal{C}^*} \GETS  {\mathcal{C}^*} \setminus \{(s,t),(s',t')\}\cup\{(s,t')\}$\;\label{alg:step3:2}
	}	

	\While{exist $(s, t), (s', t') \in {\mathcal{C}^*}$ such that $t' \leq t < s \leq s'$\label{alg:step2:1}}{
	$ {\mathcal{C}^*} \GETS  {\mathcal{C}^*} \setminus \{(s, t)\}$\;\label{alg:step2:2}
	}	
	
	\Return{$\mathcal{C}^*$;}
	\caption{{\sc Normalize}}\label{alg:normalize}
\end{algorithm}

\begin{lemma}\label{claim:normalize}
Algorithm {\sc Normalize} is correct.
\end{lemma}
\begin{proof}
Let $\mathcal{C}^*=\{(s_1,t_1),\dots,(s_h,t_h)\}$ be the set returned as output by {\sc Normalize} on $\mathcal{C}$. We first show that $\mathcal{C}^*$ is in normal form. Indeed, assume that the requests are indexed such that
$s_i\leq s_{i+1}$ for each $i\in\{1,\dots,h-1\}$. Because of step~\ref{alg:step1}, it is the case that $t_i<s_i$, for each $i\in\{1,\dots h\}$.
Assume then, for the sake of contradiction, that $t_{i^*+1} \leq s_{i^*}$ holds for an index $i^*\in\{1,\dots,h-1\}$.
Due to steps~\ref{alg:step2:1} and \ref{alg:step2:2}, we are guaranteed that $t_{i^*}<t_{i^*+1}$. But this is impossible, since the two
requests $(s_{i^*},t_{i^*})$ and  $(s_{i^*+1},t_{i^*+1})$ would have been merged in steps~\ref{alg:step3:1} and \ref{alg:step3:2}.

In order to conclude the proof, we show that every step in {\sc Normalize} preserves the optimality of the rides. Formally, let $\mathcal{\hat
C}$ be any set of requests. Let $(s,t)$ and $(s',t')$ be two requests in $\mathcal{\hat C}$. Assume that one of the following three conditions
holds:

\begin{itemize}
  \vspace{-2mm}\item[{\bf (C1)}] $s\leq t$ (see step~\ref{alg:step1});

  \vspace{-2mm}\item[{\bf (C2)}] $t<s$, $t'<s'$ and $t' \leq t \leq s' \leq s$ (see steps~\ref{alg:step3:1} and \ref{alg:step3:2});

  \vspace{-2mm}\item[{\bf (C3)}] $t'\leq t < s \leq s'$ (see steps~\ref{alg:step2:1} and \ref{alg:step2:2}).
\end{itemize}

\vspace{-2mm}Then, we claim that: $ \opt(\tuple{G,(s_0,t_0),\mathcal{\hat C}'}) \subseteq \opt(\tuple{G,(s_0,t_0),\mathcal{\hat C}})$, where
$\mathcal{\hat C}'=\mathcal{\hat C}\setminus \{(s,t)\}$ in case (1) and (3), while $\mathcal{\hat C}'=\mathcal{\hat C}\setminus
\{(s,t),(s',t')\}\cup\{(s,t')\}$ in (2).

\smallskip
{\bf (C1) and (C3)}. We show that $\feasible(\tuple{G,(s_0,t_0),\mathcal{\hat C}})=\feasible(\tuple{G,(s_0,t_0,),\mathcal{\hat C}\setminus
\{(s,t)\}})$. Indeed, this is sufficient, as the two scenarios are defined over the same weighted graph $G$.
In fact, if $\pi$ is a feasible ride for $\tuple{G,(s_0,t_0),\mathcal{\hat C}}$, then $\pi$ is clearly feasible for
$\tuple{G,(s_0,t_0,),\mathcal{\hat C}\setminus \{(s,t)\}}$, too. On the other hand, assume that $\pi=\pi_1,\dots,\pi_k$ is a feasible ride for
$\tuple{G,(s_0,t_0),\mathcal{\hat C}\setminus \{(s,t)\}}$, with $k=\len(\pi)$.
Observe that $\pi_1=s_0$ and $\pi_k=t_0$. Therefore, any request $(s,t)$ such that $s\leq t$ is trivially satisfied by $\pi$.
In order to conclude, consider now a request $(s,t)$ with $t<s$ and assume there is a request $(s',t')\in \mathcal{\hat C}$ such that $t' \leq t < s \leq s'$.
Let $i$ be the minimum time instant such that $\pi_i=s'$. Since $t'<s'$ and $\pi$ satisfies $(s',t')$, there exists a time step $i<j$ such that
$\pi_j=t'$. Given that $t' \leq t < s \leq s'$, we immediately conclude that $\pi$ satisfies $(s,t)$, too.

\smallskip
{\bf (C2)}. Recall that in this case we have $\mathcal{\hat C}'=\mathcal{\hat C}\setminus \{(s,t),(s',t')\}\cup\{(s,t')\}$. To keep notation
simple, let $\hat \R=\tuple{G,(s_0,t_0),\mathcal{\hat C}}$ and $\hat \R'=\tuple{G,(s_0,t_0),\mathcal{\hat C}'}$.  Moreover, observe that any
ride satisfying $(s,t')$ clearly satisfies $(s,t)$ and $(s',t')$. Then, we have $\feasible(\R)\subseteq \feasible(\R')$.

Assume that $\pi$ is an optimal ride for $\tuple{G,(s_0,t_0),\mathcal{\hat C}}$. If $\pi$ is feasible for $\R'$, then we can easily conclude
that $\pi$ is in $\opt(\R')$. Indeed, assume $\pi\not\in \opt(\R')$ and let $\pi'$ be a ride in $\opt(\R')$ with $w(\pi')<w(\pi)$. Since
$\feasible(\R)\subseteq \feasible(\R')$, $\pi'$ is also feasible for $\tuple{G,(s_0,t_0),\mathcal{\hat C}}$, which is impossible by the
optimality of $\pi$.
Therefore, let us consider the case where $\pi$ is not feasible for $\R'$.

Let $i$ and $i'$ (resp., $j$ and $j'$) be the minimum (resp., maximum) time steps such that $\pi_{i}=s$ and $\pi_{i'}=s'$ (resp., $\pi_{j}=t$
and $\pi_{j'}=t'$). Since $\pi$ satisfies the requests $(s',t')$ and $(s,t)$ where $t' \leq t \leq s' \leq s$, and since $s_0\leq
\l(\mathcal{R})$ and $t_0\geq \r(\mathcal{R})$, we have that
$i'\leq j' \leq j$ and $i' \leq i \leq j$.
In particular, since $\pi$ is not feasible for $\R'$, we have
$i' \leq j' < i \leq j$.
Let $i''$ be the maximum time step such that
$i \leq i'' \leq j$ with $\pi_{i''}=s'$, which exists since $\pi_{j}=t$, $\pi_{i}=s$, and $t\leq
s'\leq s$.
Let $h=\min_{i'\leq x\leq j} \pi_x$ and $H=\max_{i'\leq x\leq j} \pi_x$, and consider the ride $\hat \pi=\pi[1,i']\mapsto h\mapsto H\mapsto
\pi[i'',\len(\pi)]$. Note that $h\leq t'$ and $H\geq s$ hold. Moreover, note that $\hat \pi\preceq \pi$. By Fact~\ref{lem:preserving}, we
therefore have that $w(\hat \pi)\leq w(\pi)$.

Consider now the ride $\pi^*=\pi[1,i']\mapsto H\mapsto h\mapsto \pi[i'',\len(\pi)]$. Since $\pi_{i'}=\pi_{i''}$, we have $w(\hat
\pi)=w(\pi^*)$.
Now, observe that $\pi^*$ satisfies all requests $(s^*,t^*)$ with $h\leq t^*\leq s^*\leq H$, and of course all requests $(s,t)$ with $s\leq t$.
Consider then a request $(s^*,t^*)$ with $t^*\leq H<s^*$, which is satisfied by $\pi$. Note that $s^*\not\in \nodes(\pi[1,i'])$, by definition
of $i'$. In fact, $s^*\not\in \nodes(\pi[1,i''])$ and we conclude that $\pi[i'',\len(\pi)]$ must satisfy $(s^*,t^*)$. Therefore, $\pi^*$
satisfies $(s^*,t^*)$, too.
Similarly, consider a request $(s^*,t^*)$ with  $t^*< h\leq s^*$, which is satisfied by $\pi$. Note that $t^*\not\in \nodes(\pi[i',\len(\pi)])$
and, hence, $\pi[1,i']$ must satisfy $(s^*,t^*)$. Therefore, $\pi^*$ satisfies $(s^*,t^*)$, too.

From the above arguments, we conclude that $\pi^*$ is feasible for $\tuple{G,(s_0,t_0),\mathcal{\hat C}}$. By recalling that $w(\pi^*)=w(\hat
\pi)\leq w(\pi)$, we get that $\pi^*$ is actually an optimal ride. Moreover, $\pi^*$ satisfies $(s,t')$, and is hence a feasible ride for
$\tuple{G,(s_0,t_0),\mathcal{\hat C}'}$. Since $\feasible(\R)\subseteq \feasible(\R')$, $\pi^*$ is optimal for $\R'$.
\end{proof}

\begin{example}\label{example:line1}\em
Consider the execution of  {\sc Normalize} on the  following instance: $V=\{1, 2, 3, 4, 5, 6, 7\}$, $E=\{\{1,2\}, \{2,3\},$ $\{3,4\}, \{4,5\},$
$\{5,6\}, \{6,7\}\}$,  $w(e) = 1\,$ for every $e\in E$,  $s_0 =1$, $t_0=7$, and $\mathcal{C}=\{(2,3), (4,4), (4,2), (3,1), (2,1), (6,5),$
$(5,7)\}$.
Step~\ref{alg:step1} removes the three requests $(2,3)$, $(4,4)$, $(5,7)$, hence obtaining $\mathcal{C}^*=\{(4,2), (3,1), (2,1), (6,5)\}$.
Steps~\ref{alg:step3:1} and \ref{alg:step3:2} replace the two requests $(4,2)$ and $(3,1)$ with $(4,1)$, obtaining $\mathcal{C}^*=\{(4,1),
(2,1), (6,5)\}$. Finally, steps~\ref{alg:step2:1} and \ref{alg:step2:2} remove the request $(2,1)$. The set returned by {\sc Normalize} at
step~\ref{alg:stepbase} in {\sc RideOnPath\_Outer}  is $\mathcal{C}^*=\{(4,1), (6,5)\}$. \hfill $\lhd$
\end{example}


Step~\ref{alg:lastStep} in {\sc RideOnPath\_Outer} returns as output a ride defined on the basis of the ordering (with respect to the starting
node) of the requests in the set $\mathcal{C}^*=\{(s_1,t_1),\dots,(s_h,t_h)\}$ returned by {\sc Normalize}. In particular, the ride is obtained
by concatenating the rides connecting $s_i$ to $t_i$, incrementally from $i=1$ to $i=h$. In the proof of the following result, we shall
evidence that such a ride is an optimal ride for $\tuple{G,(s_0,t_0),\mathcal{C}^*}$ and hence, by Lemma~\ref{claim:normalize}, an optimal ride
for $\R$.

\begin{theorem}\label{thm:nomralized}
Algorithm {\sc RideOnPath\_Outer} is correct.
\end{theorem}
\begin{proof}
Consider Algorithm {\sc RideOnPath\_Outer}, by assuming $s_0\leq \l(\mathcal{R})\leq \r(\mathcal{R})\leq t_0$ (cf. Fact~\ref{fact:symmetric}).
By Lemma~\ref{claim:normalize}, we know that $\mathcal{C}^*=\{(s_1,t_1),\dots,(s_h,t_h)\}$ is in normal form. First, we show that the following
ride
$$
\pi= s_0 \mapsto s_1 \mapsto t_1 \mapsto s_2 \mapsto \ldots \mapsto s_k \mapsto t_k \mapsto t_0,
$$
\noindent which is returned by {\sc RideOnPath\_Outer}, is an optimal ride for $\tuple{G,(s_0,t_0),\mathcal{C}^*}$.

Indeed, consider a feasible ride $\hat \pi$ for $\tuple{G,(s_0,t_0),\mathcal{C}^*}$. Recall that $s_0\leq \l(\mathcal{R})\leq
\r(\mathcal{R})\leq t_0$.  For each node $v\in V$, let $\occ(v,\hat \pi)$ denote the number of occurrences of $v$ in $\hat \pi$. Then, since
$t_i<s_i$, for each $i\in\{1,\dots,k\}$, the following properties are easily seen to hold on $\hat \pi$:
(1) for each node $v\in V$ for which an index $ i$ exists such that $t_{ i}<v<s_{ i}$, $\occ(v,\hat \pi)\geq 3$;
(2) for each node $v\in V$ for which an index $ i$ exists such that $v\in \{s_{ i},t_{ i}\}$, $\occ(v,\hat \pi)\geq 2$; and,
(3) for each other node $v\in V$, $\occ(v,\hat \pi)\geq 1$ holds.
In fact, note that $\pi$ satisfies every request in $\mathcal{C}$ and that the number of occurrences of each node $v\in V$ coincides with the
corresponding lower bound stated above. Therefore, $\pi$ is optimal for $\tuple{G,(s_0,t_0),\mathcal{C}^*}$.

Given that $\pi$ is optimal for $\tuple{G,(s_0,t_0),\mathcal{C}^*}$ and is returned as output, the correctness of {\sc RideOnPath\_Outer}
eventually follows by Lemma~\ref{claim:normalize}.
\end{proof}

\begin{example}\label{example:line2}\em
Consider the instance introduced in Example \ref{example:line1}. Given the set of requests $\mathcal{C}^*=\{(4,1),  (6,5)\}$ calculated at
step~\ref{alg:stepbase} in  {\sc RideOnPath\_Outer}, the ride returned at  step~\ref{alg:lastStep} is $1 \mapsto 4 \mapsto 1 \mapsto 6 \mapsto
5 \mapsto 7.$ \hfill $\lhd$
\end{example}

\subsection{Solution to the ``inner'' case}

Let us now move to analyze the ``inner'' case, where $\{s_0,t_0\} \cap \{ v\in V \mid \l(\mathcal{R})< v < \r(\mathcal{R})\}\neq \emptyset$
holds.
Let us introduce some notation. For any feasible ride $\pi$, denote by $\leftIndex(\pi)$ (resp., $\rightIndex(\pi)$) the minimum time step $i$
such that $\pi_i=\l(\mathcal{R})$ (resp., $\pi_i=\r(\mathcal{R})$). Note that $\leftIndex(\pi)$ and $\rightIndex(\pi)$ are well defined and, in
particular, $\leftIndex(\pi)\neq \rightIndex(\pi)$ holds, since $\l(\mathcal{R}) < \r(\mathcal{R})$.
Moveover, for every pair of nodes $x,y\in V$ with $x < y$, define $\mathcal{R}(x,y)=\tuple{G,(x,y),\{ (s,t)\in\mathcal{C} \mid x\leq s,t\leq
y\}}$, that is, the scenario which inherits from $\R$ the graph $G$ and every request with both starting and terminating nodes in the interval
$\{x,...,y\}$, and where the vehicle is asked to start from $x$ and to terminate at $y$.
Notice that, by definition, the set of all nodes occurring in any optimal ride for $\mathcal{R}(x,y)$ is a subset of $\{x,...,y\}$.

\subsubsection{Canonical rides}

A crucial role in our analysis is played by the concept of canonical ride, which is illustrated below.

\begin{definition}\label{def:canonical}\em
Let $M,m\in V_{\mathcal{C}}\cup\{s_0,t_0\}$ be two nodes.  A ride $\pi^{\mathsf{c}}$ in $\R$ is said to be $(M, m)$-\emph{canonical} if
$\pi^{\mathsf{c}}=\pi'\mapsto \pi'' \mapsto \pi'''$ where
\begin{itemize}
  \item $\pi'=s_0\mapsto M \mapsto \l(\mathcal{R})\mapsto M$;

  \item  $\pi''=\left\{\begin{array}{ll}
  M\mapsto \r(\mathcal{R}) & \mbox{ if $m\leq M$}\\
  \bar \pi \mapsto \r(\mathcal{R}) & \mbox{ if $M< m$}
  \end{array}\right.$\\ where  $\bar \pi$ is an optimal ride for $\mathcal{R}(M,m)$;

  \item  $\pi'''=\r(\mathcal{R})\mapsto m \mapsto t_0$.~\hfill $\Box$
\end{itemize}
\end{definition}

\begin{figure}[t]{\hspace{-6mm}
\centering
\subfigure[$M < m$]{
 
\def\v{-.5}
\def\h{0}

\def\vA{-.5}
\def\vB{-.8}
\def\vC{-1.4}
\def\vD{-2}
\def\vY{-2.5}
\def\vX{-3}
\def\vE{-3.5}
\def\vF{-4.1}
\def\vG{-4.7}
\def\vH{-5}

\def\hA{0}
\def\hB{2}
\def\hC{4}
\def\hX{5}
\def\hRR{5.5}
\def\hY{6}
\def\hD{7}
\def\hE{9}
\def\hF{11}

\def\size{\scriptsize}

\psset{xunit=.53cm}
\psset{yunit=.53cm}

\psset{dash=3pt 3pt}

\begin{pspicture}(-0.5,-0.5)(\hF,\vH)  


\psset{linewidth=1pt}
\psset{linecolor=red}

\uput[90](\hA,\h){\size $\l$}
\psline{|-|}(\hA,\h)(\hB,\h)\uput[90](\hB,\h){\size $s_0$}
\psline{-|}(\hB,\h)(\hC,\h)\uput[90](\hC,\h){\size $M$}
\psline{-|}(\hC,\h)(\hD,\h)\uput[90](\hD,\h){\size $m$}
\psline{-|}(\hD,\h)(\hE,\h)\uput[90](\hE,\h){\size $t_0$}
\psline{-|}(\hE,\h)(\hF,\h)\uput[90](\hF,\h){\size $\r$}

\psline{|-|}(\v,\vA)(\v,\vC)\uput[180](\v,\vC){\size $\leftIndex$}
\psline{-|}(\v,\vC)(\v,\vF)\uput[180](\v,\vF){\size $\rightIndex$}
\psline{->}(\v,\vF)(\v,\vH)

\psset{linewidth=.7pt}
\psset{linecolor=black}

\psline[curvature=.7 1 0](\hB,\vA)(\hC,\vB)(\hA,\vC)(\hC,\vD)

\psline[curvature=.7 1 0](\hD,\vE)(\hF,\vF)(\hD,\vG)(\hE,\vH)

\psdots(\hB,\vA)(\hC,\vB)(\hA,\vC)(\hC,\vD)(\hD,\vE)(\hF,\vF)(\hD,\vG)(\hE,\vH)

\pscurve[dash=4pt 3pt,linestyle=dashed,curvature=.7 1 0](\hC,\vD)(\hY,\vY)(\hX,\vX)(\hD,\vE)

\uput[-90](\hRR,\vB){\size $\R(M, m)$}

\psset{linewidth=.5pt}
\psset{linecolor=green}
\psset{linestyle=dashed}

\psline{-}(\hA,\h)(\hA,\vC)
\psline{-}(\hB,\h)(\hB,\vA)

\psline{-}(\hC,\h)(\hC,\vB)
\psline{-}(\hD,\h)(\hD,\vG)
\psline{-}(\hE,\h)(\hE,\vH)
\psline{-}(\hF,\h)(\hF,\vF)

\psline{-}(\v,\vC)(\hA,\vC)
\psline{-}(\v,\vF)(\hF,\vF)

\end{pspicture}
\label{fig:A}
}
\subfigure[$m \leq M$]{
 
\def\v{-.5}
\def\h{0}

\def\vA{-.5}
\def\vB{-1.1}
\def\vC{-2}
\def\vD{-3.5}
\def\vE{-4.4}
\def\vF{-5}

\def\hA{0}
\def\hB{2.2}
\def\hC{4.4}
\def\hD{6.6}
\def\hE{8.8}
\def\hF{11}

\def\size{\scriptsize}

\psset{xunit=.53cm}
\psset{yunit=.53cm}

\psset{dash=3pt 3pt}

\begin{pspicture}(-0.5,-0.5)(\hF,\vF)  


\psset{linewidth=1pt}
\psset{linecolor=red}

\uput[90](\hA,\h){\size $\l$}
\psline{|-|}(\hA,\h)(\hB,\h)\uput[90](\hB,\h){\size $s_0$}
\psline{-|}(\hB,\h)(\hC,\h)\uput[90](\hC,\h){\size $m$}
\psline{-|}(\hC,\h)(\hD,\h)\uput[90](\hD,\h){\size $M$}
\psline{-|}(\hD,\h)(\hE,\h)\uput[90](\hE,\h){\size $t_0$}
\psline{-|}(\hE,\h)(\hF,\h)\uput[90](\hF,\h){\size $\r$}

\psline{|-|}(\v,\vA)(\v,\vC)\uput[180](\v,\vC){\size $\leftIndex$}
\psline{-|}(\v,\vC)(\v,\vD)\uput[180](\v,\vD){\size $\rightIndex$}
\psline{->}(\v,\vD)(\v,\vF)

\psset{linewidth=.7pt}
\psset{linecolor=black}

\psline[curvature=.7 1 0](\hB,\vA)(\hD,\vB)(\hA,\vC)(\hF,\vD)(\hC,\vE)(\hE,\vF)

\psdots(\hB,\vA)(\hD,\vB)(\hA,\vC)(\hF,\vD)(\hC,\vE)(\hE,\vF)

\psset{linewidth=.5pt}
\psset{linecolor=green}
\psset{linestyle=dashed}

\psline{-}(\hA,\h)(\hA,\vC)
\psline{-}(\hB,\h)(\hB,\vA)
\psline{-}(\hC,\h)(\hC,\vE)
\psline{-}(\hD,\h)(\hD,\vB)
\psline{-}(\hE,\h)(\hE,\vF)
\psline{-}(\hF,\h)(\hF,\vD)

\psline{-}(\v,\vC)(\hA,\vC)
\psline{-}(\v,\vD)(\hF,\vD)

\end{pspicture}
\label{fig:B}
}}
\caption{Example of $(M, m)$-canonical rides.}
\label{fig:canonical}
\end{figure}
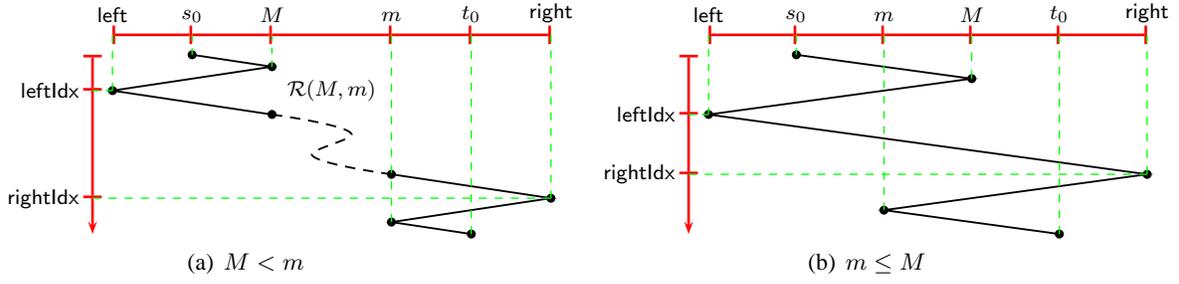

Two examples of canonical rides are in Figure~\ref{fig:canonical}. Note that if $m\leq M$ holds, we can refer without ambiguities to \emph{the}
$(M, m)$-{canonical} ride, as there is precisely one ride enjoying the properties in Definition~\ref{def:canonical}.

\begin{fact}\label{fact:CR}
If $m\leq M$, then $(M, m)$-{canonical} ride is $s_0\mapsto M\mapsto \l(\mathcal{R})\mapsto \r(\mathcal{R})\mapsto m \mapsto t_0$.
\end{fact}

Instead, whenever $m>M$, there can be more than one canonical ride. In this case, to compute a $(M, m)$-canonical ride, we need to compute an
optimal ride for $\mathcal{R}(M,m)$, which is a scenario fitting the ``outer'' case and which can be hence addressed via the {\sc
RideOnPath\_Outer} algorithm.

\smallskip

In fact, the notion of canonical ride characterizes the optimal rides for $\R$. In particular, observe that in the following result, we focus
on optimal rides $\pi^*$ such that $\leftIndex(\pi^*)< \rightIndex(\pi^*)$. Indeed, the case where $\leftIndex(\pi^*)\geq \rightIndex(\pi^*)$
will be eventually addressed by working on the symmetric scenario $\sym(\R)$, according to the approach discussed in Section~\ref{sec:outer}
(see Fact~\ref{fact:symmetric}).

\begin{theorem}\label{thm:path}
Assume that $\pi^*$ is an optimal ride with $\leftIndex(\pi^*)< \rightIndex(\pi^*)$. Then, there are two nodes $M,m\in
V_{\mathcal{C}}\cup\{s_0,t_0\}$, with $s_0\leq M$ and $m\leq t_0$, such that any $(M, m)$-canonical ride is optimal, too.
\end{theorem}

The proof of the result is rather involved, and the rest of this section is devoted to illustrate it in detail.


%
%

\smallskip

Assume that $\pi^*$ is an optimal ride such that $\leftIndex(\pi^*)< \rightIndex(\pi^*)$.
We first define a number of critical time steps and nodes of the path which are useful to analyze the properties of any optimal ride $\pi$. To
help the intuition, the reader is referred to Figure~\ref{fig1:A}.

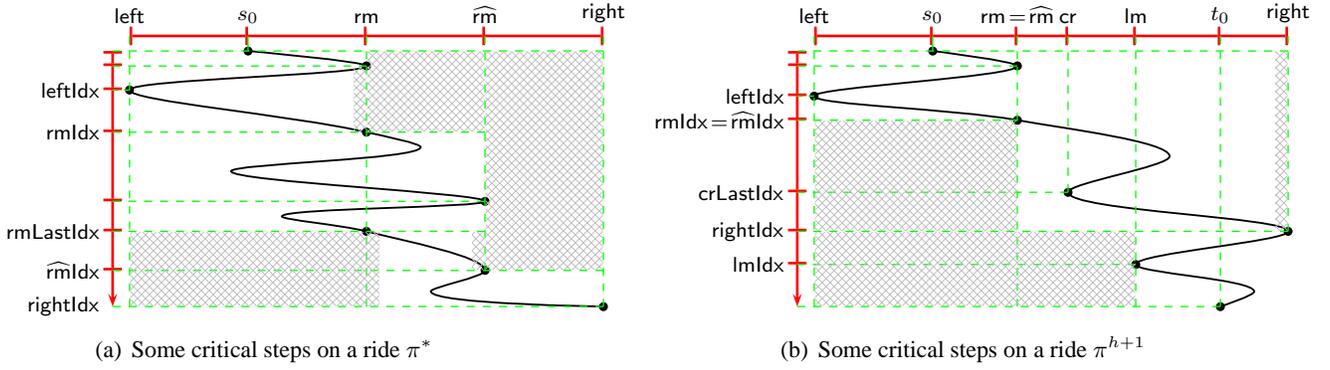
\begin{figure}[t]{\hspace{-20mm}
\centering
\subfigure[Some critical steps on a ride $\pi^*$]{
\def\v{-.5}
\def\h{0}

\def\vA{-.5}
\def\vB{-1}
\def\vC{-1.8}
\def\vD{-3.2}
\def\vXa{-3.7}
\def\vXb{-4.5}
\def\vE{-5.5}
\def\vXc{-6}
\def\vF{-6.5}
\def\vG{-7.8}
\def\vXd{-8.5}
\def\vH{-9}

\def\hA{0}
\def\hXb{3}
\def\hB{3.5}
\def\hXc{4.5}
\def\hCm{6.6}
\def\hC{7}
\def\hCp{7.4}
\def\hXa{8.6}
\def\hXd{8.9}
\def\hDm{10.1}
\def\hD{10.5}
\def\hE{14}

\def\size{\scriptsize}

\psset{xunit=.45cm}
\psset{yunit=.4cm}



\psset{dash=3pt 3pt}

\begin{pspicture}(-0.5,-0.5)(14,-9.5) 


\psset{linewidth=1pt}
\psset{linecolor=red}

\uput[90](\hA,\h){\size $\l$}
\psline{|-|}(\hA,\h)(\hB,\h)\uput[90](\hB,\h){\size $s_0$}
\psline{-|}(\hB,\h)(\hC,\h)\uput[90](\hC,\h){\size $\rm$}
\psline{-|}(\hC,\h)(\hD,\h)\uput[90](\hD,\h){\size $\brm$}
\psline{-|}(\hD,\h)(\hE,\h)\uput[90](\hE,\h){\size $\r$}

\psline{|-|}(\v,\vA)(\v,\vB)
\psline{-|}(\v,\vB)(\v,\vC)\uput[180](\v,\vC){\size $\leftIndex$}
\psline{-|}(\v,\vC)(\v,\vD)\uput[180](\v,\vD){\size $\rmIndex$}
\psline{-|}(\v,\vD)(\v,\vE)
\psline{-|}(\v,\vE)(\v,\vF)\uput[180](\v,\vF){\size $\rmLastIndex$}
\psline{-|}(\v,\vF)(\v,\vG)\uput[180](\v,\vG){\size $\brmIndex$}
\psline{->}(\v,\vG)(\v,\vH)\uput[180](\v,\vH){\size $\rightIndex$}

\psset{linewidth=.7pt}
\psset{linecolor=black}

\pscurve[curvature=.7 1 0](\hB,\vA)(\hC,\vB)(\hA,\vC)(\hC,\vD)(\hXa,\vXa)(\hXb,\vXb)(\hD,\vE)(\hXc,\vXc)(\hC,\vF)(\hD,\vG)(\hXd,\vXd)(\hE,\vH)

\psdots(\hB,\vA)(\hC,\vB)(\hA,\vC)(\hC,\vD)(\hD,\vE)(\hC,\vF)(\hD,\vG)(\hE,\vH)

\psset{linewidth=.5pt}
\psset{linecolor=green}
\psset{linestyle=dashed}

\psline{-}(\hA,\h)(\hA,\vH)
\psline{-}(\hB,\h)(\hB,\vA)
\psline{-}(\hC,\h)(\hC,\vH)
\psline{-}(\hD,\h)(\hD,\vG)
\psline{-}(\hE,\h)(\hE,\vH)

\psline{-}(\v,\vA)(\hE,\vA)
\psline{-}(\v,\vB)(\hC,\vB)
\psline{-}(\v,\vC)(\hA,\vC)
\psline{-}(\v,\vD)(\hD,\vD)
\psline{-}(\v,\vE)(\hD,\vE)
\psline{-}(\v,\vF)(\hD,\vF)

\psline{-}(\v,\vG)(\hE,\vG)
\psline{-}(\v,\vH)(\hE,\vH)

\psset{linestyle=none, fillstyle=crosshatch, hatchcolor=lightgray, hatchwidth=.2pt, hatchsep=2pt}

\psframe(\hE,\vA)(\hD,\vG)
\psframe(\hD,\vF)(\hDm,\vG)
\psframe(\hD,\vA)(\hC,\vD)
\psframe(\hC,\vB)(\hCm,\vD)
\psframe(\hA,\vH)(\hCp,\vF)

\end{pspicture}
\label{fig1:A}
}
\subfigure[Some critical steps on a ride $\pi^{h+1}$]{
 
\def\v{-.5}
\def\h{0}

\def\vA{-.5}
\def\vB{-1}
\def\vC{-2}
\def\vD{-2.8}
\def\vXa{-4.0}
\def\vE{-5.2}
\def\vF{-6.5}
\def\vG{-7.6}
\def\vXb{-8.5}
\def\vH{-9}

\def\hA{0}
\def\hB{3.5}
\def\hC{6}
\def\hD{7.5}
\def\hE{9.5}
\def\hXa{10.5}
\def\hF{12}
\def\hXb{13}
\def\hGm{13.6}
\def\hG{14}

\def\size{\scriptsize}

\psset{xunit=.45cm}
\psset{yunit=.4cm}



\psset{dash=3pt 3pt}

\begin{pspicture}(-0.5,-0.5)(14,-9.5)  


\psset{linewidth=1pt}
\psset{linecolor=red}

\uput[90](\hA,\h){\size $\l$}
\psline{|-|}(\hA,\h)(\hB,\h)\uput[90](\hB,\h){\size $s_0$}
\psline{-|}(\hB,\h)(\hC,\h)\uput[90](\hC,\h){\size $\rm\!=\!\brm$}
\psline{-|}(\hC,\h)(\hD,\h)\uput[90](\hD,\h){\size $\crp$}
\psline{-|}(\hD,\h)(\hE,\h)\uput[90](\hE,\h){\size $\lm$}
\psline{-|}(\hE,\h)(\hF,\h)\uput[90](\hF,\h){\size $t_0$}
\psline{-|}(\hF,\h)(\hG,\h)\uput[90](\hG,\h){\size $\r$}

\psline{|-|}(\v,\vA)(\v,\vB)
\psline{-|}(\v,\vB)(\v,\vC)\uput[180](\v,\vC){\size $\leftIndex$}
\psline{-|}(\v,\vC)(\v,\vD)\uput[180](\v,\vD){\size $\rmIndex\!=\!\brmIndex$}
\psline{-|}(\v,\vD)(\v,\vE)\uput[180](\v,\vE){\size $\crpIndex$}
\psline{-|}(\v,\vE)(\v,\vF)\uput[180](\v,\vF){\size $\rightIndex$}
\psline{-|}(\v,\vF)(\v,\vG)\uput[180](\v,\vG){\size $\lmIndex$}
\psline{->}(\v,\vG)(\v,\vH)

\psset{linewidth=.7pt}
\psset{linecolor=black}

\pscurve[curvature=.7 1 0](\hB,\vA)(\hC,\vB)(\hA,\vC)(\hC,\vD)(\hXa,\vXa)(\hD,\vE)(\hG,\vF)(\hE,\vG)(\hXb,\vXb)(\hF,\vH)

\psdots(\hB,\vA)(\hC,\vB)(\hA,\vC)(\hC,\vD)(\hD,\vE)(\hG,\vF)(\hE,\vG)(\hF,\vH)

\psset{linewidth=.5pt}
\psset{linecolor=green}
\psset{linestyle=dashed}

\psline{-}(\hA,\h)(\hA,\vH)
\psline{-}(\hB,\h)(\hB,\vA)
\psline{-}(\hC,\h)(\hC,\vH)
\psline{-}(\hD,\h)(\hD,\vE)
\psline{-}(\hE,\h)(\hE,\vH)
\psline{-}(\hF,\h)(\hF,\vH)
\psline{-}(\hG,\h)(\hG,\vF)

\psline{-}(\v,\vA)(\hG,\vA)
\psline{-}(\v,\vB)(\hC,\vB)
\psline{-}(\v,\vC)(\hA,\vC)
\psline{-}(\v,\vD)(\hC,\vD)
\psline{-}(\v,\vE)(\hD,\vE)
\psline{-}(\v,\vF)(\hG,\vF)
\psline{-}(\v,\vG)(\hE,\vG)
\psline{-}(\v,\vH)(\hF,\vH)

\psset{linestyle=none, fillstyle=crosshatch, hatchcolor=lightgray, hatchwidth=.2pt, hatchsep=2pt}

\psframe(\hE,\vF)(\hA,\vH)
\psframe(\hC,\vD)(\hA,\vF)
\psframe(\hG,\vA)(\hGm,\vF)

\end{pspicture}
\label{fig1:B}
}}\vspace{-2mm}
\caption{Some critical steps of any feasible ride on a path.  The gray areas denote the space that no feasible ride can cross for a given time interval.}\vspace{2mm}
\label{fig2:canonical}
\end{figure}


Let $\rm(\pi)=\max_{1\leq i\leq \leftIndex(\pi)} \pi_i$. Note that $\rm(\pi)< \r(\mathcal{R})$ necessarily holds.
Let $\rmIndex(\pi)$ be the minimum time step $i\geq \leftIndex(\mathcal{R})$ such that $\pi_i=\rm(\pi)$.
Note that that $\rmIndex(\pi)$ is well defined, because $\leftIndex(\pi)< \rightIndex(\pi)$ and, hence, the ride $\pi$ has to cross the node
$\rm(\pi)$ at least once between the time step $\leftIndex(\pi)$ and the time step $\rightIndex(\pi)$. In fact, it actually holds that
$\rmIndex(\pi)<\rightIndex(\pi)$, since  $\rm(\pi)< \r(\mathcal{R})$.
Then, define $\rmLastIndex(\pi)$ as the maximum time step $i\leq \rightIndex(\pi)$ such that $\pi_i=\rm(\pi)$.
Note that $\rmLastIndex(\pi)$ coincides with $\rmIndex(\pi)$ if, and only if, there is no time step $i$ such that $\rmIndex(\pi)<i\leq
\rightIndex(\pi)$ with $\pi_i=\rm(\pi)$. Again, observe that $\rmLastIndex(\pi)<\rightIndex(\pi)$ holds.

Now, define $\brm(\pi)=\max_{\rmIndex(\pi)\leq i\leq \rmLastIndex(\pi)} \pi_i$. Since $\rmLastIndex(\pi)<\rightIndex(\pi)$  and since
$\rightIndex(\pi)$ is the minimum time step where the ride reaches the extreme node $\r(\mathcal{R})$, we have that
$\brm(\pi)<\r(\mathcal{R})$. Moreover, $\brm(\pi)\geq \rm(\pi)$ clearly holds. Therefore, there is some time step between $\rmLastIndex(\pi)$
and $\rightIndex(\pi)$ where $\pi$ crosses $\brm(\pi)$. So, we can define $\brmIndex(\pi)$ as the minimum index $i\geq \rmLastIndex(\pi)$ such
that $\pi_i=\brm(\pi)$, by noticing that $\brmIndex(\pi)<\rightIndex(\pi)$ holds.

Eventually, define also $\lm(\pi)=\min_{\rightIndex(\pi)\leq i\leq \len(\pi)} \pi_i$.

\begin{lemma}\label{fact:VC}
Assume there is an optimal ride $\pi'$ such that $\leftIndex(\pi')< \rightIndex(\pi')$. Then, there is an optimal ride $\pi$ such that
$\leftIndex(\pi)< \rightIndex(\pi)$ and where $\lm(\pi)$, $\brm(\pi)$ and $\rm(\pi)$ belong to the set $V_{\mathcal{C}}\cup\{s_0,t_0\}$.
\end{lemma}
\begin{proof}
We illustrate the case of $\rm$, since a similar line of reasoning applies to $\lm$ and $\brm$.
Assume that $\rm(\pi')\not\in V_{\mathcal{C}} \cup\{s_0,t_0\}$. Consider the succession of rides $\pi^j$, with $j\geq 0$, built as follows.
Initially, i.e., for $j=0$, we set $\pi^j=\pi'$. Consider any time step $i$ such that $1\leq i\leq \leftIndex(\pi^j)$ and $\pi^j_i=\rm(\pi^j)$.
Note that $1<i<\leftIndex(\pi^j)$ actually holds, since $s_0\neq \rm(\pi^j)$ and $\leftIndex(\pi^j)< \rightIndex(\pi^j)$. Consider then the
ride $\pi^{j+1}=\pi^j[1,i-1]\mapsto \pi^j[i+1,\len(\pi)]$, and note that $\pi^{j+1}\preceq \pi^j$ and $\len(\pi^{j+1})<\len(\pi^j)$.
Since $\pi^j_i\not \in V_{\mathcal{C}}$, we therefore have that $\pi^{j+1}$ is optimal too, because of Fact~\ref{lem:preserving}. If
$\rm(\pi^{j+1})\in V_{\mathcal{C}}\cup\{s_0,t_0\}$, then we have concluded. Otherwise, we can repeat this method over $\pi^{j+1}$ by noticing
that $s_0\leq\rm(\pi^{j+1})\leq \rm(\pi^j)$ and $\leftIndex(\pi^{j+1})<\leftIndex(\pi^j)$. Therefore, the process will eventually converge to
an optimal ride $\pi$ such that  $\rm(\pi)$ belongs to the set $V_{\mathcal{C}}$ or coincides with $s_0$.
\end{proof}

Let us now start by analyzing the properties of the optimal rides.

\begin{lemma}\label{lem:phase1}
Assume there is an optimal ride $\pi \in \opt(\mathcal{R})$ such that $\leftIndex(\pi)< \rightIndex(\pi)$. Then, the following ride is optimal,
too:
\begin{eqnarray}\label{eqn:phase1}
s_0 \mapsto \brm(\pi) \mapsto \l(\mathcal{R}) \mapsto \brm(\pi) \mapsto \pi[\brmIndex(\pi),\len(\pi)].
\end{eqnarray}
\end{lemma}
\begin{proof}
Let $\hat \pi \mapsto \pi[\brmIndex(\pi),\len(\pi)]$ be the ride where $\hat \pi=s_0 \mapsto \brm(\pi) \mapsto \l(\mathcal{R}) \mapsto
\brm(\pi)$. Observe that $w(\hat \pi)\leq w(\pi[1,\brmIndex(\pi)])$. Moreover, we shall show that for each request $(s,t)\in \mathcal{C}$,
$\hat \pi \mapsto \pi[\brmIndex(\pi),\len(\pi)]$ satisfies $(s,t)$.
This will immediately imply that $\hat \pi \mapsto \pi[\brmIndex(\pi),\len(\pi)]$ is an optimal ride, too.

Recall first that, since $\pi$ is a feasible ride, for each request $(s,t)$, there are two time steps $i$ and $i'$ such that $1\leq i\leq
i'\leq \len(\pi)$, $\pi_i=s$ and $\pi_{i'}=t$.
Now, if $i\geq \brmIndex(\pi)$, then $\pi[\brmIndex(\pi),\len(\pi)]$ satisfies $(s,t)$; hence, $\hat \pi \mapsto \pi[\brmIndex(\pi),\len(\pi)]$
satisfies $(s,t)$, too.
Assume then that $i'\leq \brmIndex(\pi)$, and let us distinguish the following two cases: \emph{(i)} if $s\leq t$, then $\l(\mathcal{R})
\mapsto \brm(\pi)$ satisfies $(s,t)$; \emph{(ii)} otherwise, i.e., if $s>t$, then $\brm(\pi) \mapsto \l(\mathcal{R})$ satisfies $(s,t)$. In
both cases, we can conclude that $\hat \pi\mapsto \pi[\brmIndex(\pi),\len(\pi)]$ satisfies $(s,t)$, too.
Finally, assume that $i< \brmIndex(\pi)<i'$. In this case, $s$ is in $\nodes(\hat \pi)=\nodes(\pi[1,\brmIndex(\pi)])$, while $t$ is in
$\nodes(\pi[\brmIndex(\pi),\len(\pi)])$. Thus, $\hat \pi \mapsto \pi[\brmIndex(\pi),\len(\pi)]$ satisfies $(s,t)$.
\end{proof}

Consider now the optimal ride $\pi^*$, and the succession of optimal rides $\pi^j$, with $j\geq 0$, obtained by repeatedly applying
Lemma~\ref{lem:phase1}. First, we set $\pi^0=\pi^*$. Then, for each $j\geq 0$, we define $\pi^{j+1}$ as the optimal ride having the form:
$$
s_0 \mapsto \brm(\pi^j) \mapsto \l(\mathcal{R}) \mapsto \brm(\pi^j) \mapsto \pi^j[\brmIndex(\pi^j),\len(\pi^j)].
$$

In the above succession, there must exists an optimal ride $\pi^h$, with $h\geq 0$, such that $\brm(\pi^{h})=\rm(\pi^{h})$. Indeed, note that
$\rm(\pi^{j+1})= \brm(\pi^j)$ holds, for each $j\geq 0$, and we know that, for any optimal ride $\pi$, $\rm(\pi)\leq
\brm(\pi)<\r(\mathcal{R})$.

For this optimal ride $\pi^h$, we have that $\rmLastIndex(\pi^{h})=\brmIndex(\pi^h)$, by definition of these two time steps. Therefore,
$$
\pi^{h+1}=s_0 \mapsto \rm(\pi^h) \mapsto \l(\mathcal{R}) \mapsto \rm(\pi^h) \mapsto \pi^h[\rmLastIndex(\pi^h),\len(\pi^h)].
$$
For the subsequent analysis, we shall write $\pi^{h+1}=\pi' \mapsto \hat \pi''\mapsto \hat \pi'''$ where:
\begin{itemize}
  \item $\pi'=s_0 \mapsto \rm(\pi^h) \mapsto \l(\mathcal{R}) \mapsto \rm(\pi^h)$;

  \item $\hat \pi''=\pi^h[\rmLastIndex(\pi^h),\rightIndex(\pi^h)]$; and

  \item $\hat \pi'''=\pi^h[\rightIndex(\pi^h),\len(\pi^h)]$.
\end{itemize}

Figure~\ref{fig1:B} reports an illustration of the result discussed below.


\begin{lemma}\label{lem:properties}
The following properties hold on $\pi^{h+1}=\pi' \mapsto \hat \pi''\mapsto \hat \pi'''$:
\begin{enumerate}
  \item[(1)] $\hat \pi''_{\len(\hat \pi'')}=\r(\mathcal{R})$; there is no node $v\in \nodes(\hat \pi'')$ such that $v< \rm(\pi^h)$; and,
      for each time step $i$ with $1\leq i<\len(\hat \pi'')$, $\hat \pi''_{i} \neq \r(\mathcal{R})$;
  \item[(2)] for each node $v\in \nodes(\hat \pi''')$, $v\geq \lm(\pi^{h+1})$;
  \item[(3)] there is no request $(s,t)\in\mathcal{C}$ such that $t<\lm(\pi^{h+1})$, $t<\rm(\pi^{h})$, and $\rm(\pi^h)< s$.
\end{enumerate}
\end{lemma}
\begin{proof}
Property \emph{(1)} is immediate since $\hat \pi''=\pi^h[\rmLastIndex(\pi^h),\rightIndex(\pi^h)]$, and given the definition of the time steps
$\rmLastIndex(\pi^h)$ and $\rightIndex(\pi^h)$.

Similarly, property \emph{(2)} holds because $\hat \pi'''=\pi^h[\rightIndex(\pi^h),\len(\pi^h)]$ and given the definition of $\lm(\pi^{h+1})$.

Concerning property \emph{(3)}, assume for the sake of contradiction that $(s,t)$ is a request such that $t<\lm(\pi^{h+1})$, $t<\rm(\pi^{h})$,
and $\rm(\pi^h)\leq s$.
By property \emph{(1)} and property \emph{(2)}, we have that $t\not\in \nodes(\hat \pi''\mapsto \hat \pi''')$.
However, for each node $v\in \nodes(\pi')$, it holds that $v\leq \rm(\pi^h)$. Given that $s >\rm(\pi^h)$, this entails that $s\not\in
\nodes(\pi')$. Combined with the fact that $t\not\in \nodes(\hat \pi''\mapsto \hat \pi''')$, then we derive that $\pi^{h+1}$ does not satisfy
$(s,t)$, which is impossible.
\end{proof}

Armed with the above properties, we can now analyze the form of the rides $\hat \pi''$ and $\hat \pi'''$. We start with the case where
$\lm(\pi^{h+1})< \rm(\pi^h)$.

\begin{lemma}\label{lem:phase2}
If $\lm(\pi^{h+1})< \rm(\pi^h)$, then the ride $\pi'\mapsto \pi''\mapsto \pi'''$ is optimal, where $\pi''=\rm(\pi^h)\mapsto \r(\mathcal{R})$
and $\pi'''=\r(\mathcal{R})\mapsto \lm(\pi^{h+1})\mapsto t_0$.
\end{lemma}
\begin{proof} Define $\pi''=\rm(\pi^h)\mapsto \r(\mathcal{R})$ and $\pi'''=\r(\mathcal{R})\mapsto \lm(\pi^{h+1})\mapsto t_0$. We have to show that
$\pi'\mapsto \pi''\mapsto \pi'''$ is an optimal ride. In fact, it is immediate to check that $\pi'\mapsto \pi''\mapsto \pi'''\preceq
\pi^{h+1}$. Therefore, after Lemma~\ref{lem:preserving}, we have just to show that, for each request $(s,t)\in \mathcal{C}$, $\pi'\mapsto
\pi''\mapsto \pi'''$ satisfies $(s,t)$.

Let $(s,t)$ be a request in $\mathcal{C}$. If $s\leq t$, then ride $\l(\mathcal{R}) \mapsto \rm(\pi^h) \mapsto \r(\mathcal{R})$ trivially
satisfies $(s,t)$. Then, consider the case where $s>t$, and let us distinguish the following two possibilities. If $t\geq \lm(\pi^{h+1})$, then
$\pi'''$ satisfies $(s,t)$. Instead, if $t<\lm(\pi^{h+1})$, then we know that $t<\lm(\pi^{h+1})< \rm(\pi^h)$ also holds. Therefore, we are in
the position of applying  property \emph{(3)} in Lemma~\ref{lem:properties}, by concluding that $s< \rm(\pi^h)$ holds. So, $\rm(\pi^h)\mapsto
\l(\mathcal{R})$ satisfies $(s,t)$.
\end{proof}

Note that, by setting $M=\rm(\pi^{h})$ and $m=\lm(\pi^{h+1})$, if $m< M$ holds (and actually even if $m=M$), then the ride in
Lemma~\ref{lem:phase2} is canonical w.r.t.~$M$ and $m$. In particular, we know that we can focus, w.l.o.g., on the case where $M$ and $m$
belongs to $V_\mathcal{C}\cup \{s_0,t_0\}$ (cf. Lemma~\ref{fact:VC}). Hence, in order to complete the proof of Claim~\ref{thm:path}, we have
now to analyze the case where $\lm(\pi^{h+1})\geq \rm(\pi^h)$.

Consider the optimal ride $\pi^{h+1}=\pi'\mapsto \hat \pi''\mapsto \hat \pi'''$, by assuming that $\lm(\pi^{h+1})\geq \rm(\pi^h)$. Moreover,
consider the notion of \emph{critical} request defined inductively as follows:
First, we say that any request $(s,t)\in \mathcal{C}$ such that $t<\lm(\pi^{h+1})\leq s$ and $s>\rm(\pi^{h})$ is critical. Then, in general, a
request $(s,t)$ is critical if $t< s$ and there is a critical request $(s',t')$ with  $t<t'\leq s$ and $s> \rm(\pi^{h})$.

Let $\mathcal{S}$ be the set of all critical requests in $\mathcal{C}$, and whenever $\mathcal{S}\neq \emptyset$, let
$\crp(\pi^{h+1})=\min_{(s,t)\in \mathcal{S}} t$.
We claim that if $\lm(\pi^{h+1})\geq \rm(\pi^h)$, then $\crp(\pi^{h+1})\geq \rm(\pi^h)$. Indeed, assume by contradiction that there is a
request $(s,t)$ such that $s> \rm(\pi^{h})$ and $t<\rm(\pi^h)$. Then, we also have that $t<\lm(\pi^{h+1})$. Hence, we get a contradiction with
property \emph{(3)} in Lemma~\ref{lem:properties}.
For uniformity, if $\mathcal{S}=\emptyset$, then we define $\crp(\pi^{h+1})=\lm(\pi^h)$ (so we again have $\crp(\pi^{h+1})\geq \rm(\pi^h)$).
Then, let $\crpIndex(\pi^{h+1})$ (resp.,  $\crpFirstIndex(\pi^{h+1})$) be the maximum time step $i\leq \rightIndex(\pi^h)$ (resp., minimum time
step $i\geq \rmLastIndex(\pi^h)$) such that $\pi^{h+1}_i=\crp(\pi^{h+1})$.

\begin{lemma}\label{lem:property2}
If $\lm(\pi^{h+1})\geq \rm(\pi^h)$, then the following properties hold:
\begin{itemize}
  \item[(1)] there is no request $(s,t)$ such that $t<\crp(\pi^{h+1})< s$;

  \item[(2)] there is no request $(s,t)$ such that $t<\rm(\pi^{h})< s$;

  \item[(3)] $\pi^h[\rmLastIndex(\pi^h),\crpIndex(\pi^{h+1})]$ satisfies each request $(s,t)$ such that $\rm(\pi^h)\leq s\leq
      \crp(\pi^{h+1})$ and $\rm(\pi^h)\leq t \leq \crp(\pi^{h+1})$.
\end{itemize}
\end{lemma}
\begin{proof}
By definition of $\crp(\pi^{h+1})$, there is no request $(s,t)$ such that $t<\crp(\pi^{h+1})\leq s$, thereby trivially implying \emph{(1)}.

Concerning \emph{(2)}, assume by contradiction that $(s,t)$ is such that $t<\rm(\pi^{h})< s$. Then, $t<\lm(\pi^{h+1})$ would hold. But, this is
impossible by property \emph{(3)} in Lemma~\ref{lem:properties}.

Finally, consider a request $(s,t)$ such that $\rm(\pi^h)\leq s\leq \crp(\pi^{h+1})$ and $\rm(\pi^h)\leq t \leq \crp(\pi^{h+1})$. We know that
$\pi'\mapsto \hat \pi''\mapsto \hat \pi'''$ satisfies $(s,t)$. By the properties in Lemma~\ref{lem:properties} and given that
$\crp(\pi^{h+1})\leq \lm(\pi^{h+1})$, we can see that $\pi^h[\rmLastIndex(\pi^h),\crpIndex(\pi^{h+1})]$ satisfies $(s,t)$.
\end{proof}

With the above ingredients, we can now further explore the form of $\pi^{h+1}$.

\begin{lemma}\label{lem:phase3}
If $\lm(\pi^{h+1})\geq \rm(\pi^h)$ and $\pi^\diamond$ is an optimal ride for $\mathcal{R}(\rm(\pi^h),\crp(\pi^{h+1}))$, then the ride
$\pi'\mapsto {\pi}''\mapsto \pi'''$ is optimal, where
\begin{itemize}
  \item ${\pi}''=\pi^\diamond \mapsto \r(\mathcal{R})$, and

    \item $\pi'''=\r(\mathcal{R})\mapsto \crp(\pi^{h+1})\mapsto t_0$.
\end{itemize}
\end{lemma}
\begin{proof}
Recall that $\hat \pi''=\pi^h[\rmLastIndex(\pi^h),\rightIndex(\pi^h)]$. Let $(s,t)$ be any critical request. Then, $s> t$ and $s>
\rm(\pi^{h})$. In fact, we know that $t\geq \crp(\pi^{h+1})$ and, hence, $t\geq \rm(\pi^h)$. Moreover, $t<\lm(\pi^{h+1})$ holds. Because of
property \emph{(2)} and property \emph{(3)} in Lemma~\ref{lem:properties} and given the form of $\pi'$, we clearly have that $\hat \pi''$ must
satisfy $(s,t)$.
Therefore, we have that $s\mapsto t \preceq \hat \pi''$ holds, for each critical request $(s,t)$. If $\mathcal{S}\neq \emptyset$, let $\hat
s=\max_{(s,t)\in \mathcal{S}}s$. Otherwise, let $\bar s=\lm(\pi^{h+1})=\crp(\pi^{h+1})$. Note that $\hat s\geq \lm(\pi^{h+1})$ and that $\hat
s\mapsto \crp(\pi^{h+1})\preceq \pi^h[\rmLastIndex(\pi^h),\rightIndex(\pi^h)]$.

Consider now the ride $\hat \pi^\diamond$ derived from $\pi^h[\rmLastIndex(\pi^h),\crpIndex(\pi^{h+1})]$ by eliminating all nodes $v$ such that
$v>\crp(\pi^{h+1})$.
By putting it together the above observation, Lemma~\ref{lem:property2}, and Lemma~\ref{lem:properties}, we conclude that the ride $\pi'\mapsto
\pi^\circ$, where $\pi^\circ=\hat \pi^\diamond\mapsto \hat s\mapsto \crp(\pi^{h+1})\mapsto \r(\mathcal{R})\mapsto \lm(\pi^{h+1}) \mapsto t_0$
is feasible and that $\pi^\circ\preceq \hat \pi''\mapsto \hat \pi'''$.
Moreover, note that $w(\hat \pi^\diamond\mapsto \r(\R)\mapsto \pi''')\leq w(\pi^\circ)$. So, we will show that $\pi'\mapsto \hat
\pi^\diamond\mapsto \r(\R)\mapsto \pi'''$ is a an optimal ride, by just evidencing that it satisfies every request $(s,t)\in \mathcal{C}$.

Let $(s,t)$ be a request. If $s\leq t$, then trivially $\pi'\mapsto \pi^\diamond\mapsto \r(\R)\mapsto \pi'''$ satisfies $(s,t)$. Consider then
the case where $s>t$. Because of the properties \emph{(1)} and \emph{(2)}  in Lemma~\ref{lem:property2}, there are actually three possible
cases. First, we might have that $s\leq \rm(\pi^h)$, and hence $\pi'$ satisfies $(s,t)$. Second, we might have that $t\geq \crp(\pi^{h+1})$,
and hence $\r(\mathcal{R})\mapsto \crp(\pi^{h+1})$ satisfies $(s,t)$. Finally, we might have that $\rm(\pi^h)\leq s\leq \crp(\pi^{h+1})$ and
$\rm(\pi^h)\leq t \leq \crp(\pi^{h+1})$. In this case, $\pi^h[\rmLastIndex(\pi^h),\crpIndex(\pi^{h+1})]$ satisfies $(s,t)$, by property
\emph{(3)} in Lemma~\ref{lem:property2}. Then, by construction and Lemma~\ref{lem:property2}, $\hat \pi^\diamond$ satisfies $(s,t)$, too.

Finally, observe that for each $v\in \nodes(\hat \pi^\diamond)$, $\rm(\pi^h)\leq v\leq \crp(\pi^{h+1})$ holds. Therefore, $\nodes({\pi}'')\cap
\nodes(\pi')=\{\rm(\pi^h)\}$ and $\nodes({\pi}'')\cap \nodes(\pi''')=\{\crp(\pi^{h+1})\}$. Because of the optimality of $\pi'\mapsto
{\pi}''\mapsto \pi'''$, we then conclude that $\hat \pi^\diamond$ is an optimal ride for $\mathcal{R}(\rm(\pi^h),\crp(\pi^{h+1}))$. In fact,
the result holds for any optimal ride $\pi^\diamond$ for $\mathcal{R}(\rm(\pi^h),\crp(\pi^{h+1}))$ used in place of $\hat \pi^\diamond$.
\end{proof}

The proof of Theorem~\ref{thm:path} is now concluded by setting $m=\crp(\pi^{h+1})$ and $M=\rm(\pi^{h})$, and observing that $M\geq m$. Indeed,
in this case, the optimal ride defined by Lemma~\ref{lem:phase3} is canonical w.r.t.~$M$ and $m$. In particular, note that for $M=m$, the ride
coincides with the one in Lemma~\ref{lem:phase2} (when $\rm(\pi^{h})=\lm(\pi^{h+1})$).

\subsubsection{An algorithm for the ``inner'' case}

It is not difficult to see that the result in Theorem~\ref{thm:path} immediately provides us with an algorithm to compute an optimal ride,
which is based on exhaustively enumerating all possible pairs $M,m$ of elements, by computing the associated canonical ride for each of them
(either by exploiting Fact~\ref{fact:CR} if $m\leq M$, or using the {\sc RideOnPath\_Outer} algorithm on $\mathcal{R}(M,m)$ of $m>M$), and by
eventually returning the feasible one having minimum cost. Actually, in order to deal with the case where all optimal rides $\pi^*$ are such
that $\leftIndex(\pi^*)> \rightIndex(\pi^*)$, we can just apply the approach over the symmetric scenario $\sym(\R)$ too (see
Fact~\ref{fact:symmetric}), and return the best over the rides computed for $\R$ and $\sym(\R)$.

Note that the approach sketched above requires the enumeration of $|V_{\mathcal{C}}|^2$ canonical rides.
However, as we shall see in the reminder of this section, we are actually able to do better than a na\"ive enumeration over all pairs of $M$
and $m$.
To this end, we explore the properties enjoyed by canonical rides that are optimal.

We start by observing that whenever $M< m$ holds in Theorem~\ref{thm:path}, then an optimal canonical ride is determined via simple expressions
that can be calculated efficiently.

\begin{theorem}\label{thm:uno}
Assume that there are two nodes $M,m\in V_{\mathcal{C}}\cup\{s_0,t_0\}$, with $s_0\leq M$, $m\leq t_0$ and $M < m$, such that a $(M,
m)$-canonical ride is an optimal ride. Consider the two  sets {\small
\begin{eqnarray*}
\hat X  & = & \{ x\in \{s_0\}\cup V_\C  \mid x\geq s_0\ \wedge \nexists (s,t)\in \mathcal{C}\mbox{ with }   t \leq x < s\},\\
\hat Y  & = & \{ y\in \{t_0\}\cup V_\C  \mid y\leq t_0\ \wedge  \nexists (s,t)\in \mathcal{C}\mbox{ with }   t < y \leq s\}.
\end{eqnarray*}
}
It holds that $\hat X \neq \emptyset$ and  $\hat Y \neq \emptyset$.
Moreover, let $$\hat M=\min_{\hat x \in \hat X} \hat x  \mbox{\;\;\; and\;\;\;} \hat m=\max_{\hat y \in \hat Y}\hat y,$$ then $s_0\leq \hat M$, $\hat m\leq t_0$, $\hat M < \hat m$ and any $(\hat M, \hat m)$-canonical ride  is an optimal ride, too.
\end{theorem}
\begin{proof}
Let $\pi^{\mathsf{c}}$ be a $(M, m)$-canonical ride that is optimal.
According to Definition~\ref{def:canonical},  since $M < m$, $\pi^{\mathsf{c}}$ has the form $\pi'\mapsto\pi''\mapsto\pi'''$ where:
$\pi'=s_0\mapsto M\mapsto \l(\R)\mapsto M$, $\pi''=\bar \pi \mapsto \r(\mathcal{R})$ where $\bar \pi$ is an optimal ride for $\mathcal{R}(M,m)$,  and  $\hat \pi'''=\r(\mathcal{R})\mapsto m \mapsto t_0$.
Note that there is no request $(s,t)$ in $\mathcal{C}$ such that $t \leq M < s$.  Indeed, let us assume, by the way of contradiction, that such
request exists.  Note that, from the definition of $\pi^{\mathsf{c}}$, there is no pair of time steps $i$ and $i'$ such that $1\leq i\leq
i'\leq \len(\pi^{\mathsf{c}})$, with $\pi^{\mathsf{c}}_{i} < M$ and $M \leq \pi^{\mathsf{c}}_{i'}$. This implies that $\pi^{\mathsf{c}}$ does
not satisfies $(s,t)$, hence contradicting the feasibility of $\pi^{\mathsf{c}}$.
As there is no request $(s,t)$ with $t \leq M < s$, we have that $M$ belongs to $\hat X$, and hence $\hat X \neq \emptyset$. By similar
arguments, we can show that $m$ belongs to $\hat Y$, and hence $\hat Y \neq \emptyset$.

Let us prove now the next statements. Note that $s_0\leq \hat M$ and $\hat m\leq t_0$ follow directly from the definition of  $\hat X$ and
$\hat Y$, respectively. In order to show that $\hat M\leq \hat m$, we exploit the fact that $M \in \hat X$ and $m\in \hat Y$. Indeed, since
$\hat M$, by definition, is the smallest element in $\hat X$,  we get that $\hat M\leq M$ holds. By similar arguments, we can derive that
$m\leq \hat m$ holds. Since from the hypothesis $M < m$, by combining the previous two inequalities, we finally get that  $\hat M < \hat m$
and, more precisely, $\hat M \leq M < m \leq \hat m$.

It remains to show that any $(\hat M, \hat m)$-canonical ride is optimal. Let us consider a  $(\hat M, \hat m)$-canonical ride  $\hat
\pi^{\mathsf{c}}$. According to Definition~\ref{def:canonical}, since  $\hat M\leq \hat m$, $\hat \pi^{\mathsf{c}}$ has the form  $\hat
\pi'\mapsto \hat \pi''\mapsto \hat \pi'''$ where:  $\hat \pi'=s_0\mapsto \hat M\mapsto \l(\R)\mapsto \hat M$, $\hat \pi''=\bar{\bar \pi}
\mapsto \r(\R)$ where $\bar{\bar \pi}$ is an optimal ride for $\R(\hat M,\hat m)$,  and $\hat \pi'''=\r(\mathcal{R})\mapsto \hat m \mapsto
t_0$.
Let us show now that $\hat \pi^{\mathsf{c}}$ is feasible. Indeed, consider any request $(s,t)\in \mathcal{C}$. In the case where $s\leq t$, the request is
satisfied by $\l(\R)\mapsto \r(\R)$, and hence by $\hat \pi^{\mathsf{c}}$.
Consider then the case where $t<s$.  Since $\hat M \in \hat X$ and $\hat m \in \hat Y$, it is not possible that $t\leq \hat M<s$ and  $t< \hat m\leq s$.
If $s\leq \hat M$, then $(s,t)$ is satisfied by $\hat \pi'$; if $\hat M\leq t<s\leq \hat m$, then $(s,t)$ is satisfied by $\hat \pi''$; and,
finally, if $\hat m\leq t$, then $(s,t)$ is satisfied by $\hat \pi'''$. So, in all the possible cases, $(s,t)$ is satisfied by $\hat \pi^{\mathsf{c}}$,  which implies that the canonical ride $\hat \pi^{\mathsf{c}}$ is a feasible ride.
In order to prove that $\hat \pi^{\mathsf{c}}$ is also optimal, we compare  the cost of $\hat \pi^{\mathsf{c}}$ with the cost of the optimal ride $\pi^{\mathsf{c}}$.
Let us recall that $\hat M \leq M < m \leq \hat m$.
Consider the ride $\hat \pi=\hat \pi'\mapsto \ddot \pi''\mapsto \hat \pi'''$, where $\hat \pi'$ and $\hat \pi'''$ are the sub-rides defined for $\hat \pi^{\mathsf{c}}$, and
$$\ddot \pi''= \hat M\mapsto M \mapsto \hat M \mapsto {\bar \pi} \mapsto \hat m \mapsto m \mapsto \hat m \mapsto \r(\mathcal{R}),$$
where ${\bar \pi}$ is an optimal ride for $\mathcal{R}(M,  m)$.
Note that,  if in $\ddot \pi''$ we replace  $\hat M\mapsto M \mapsto \hat M \mapsto {\bar \pi} \mapsto \hat m \mapsto m \mapsto \hat m$ with $ \bar{\bar \pi}$, i.e., the optimal ride for  $\mathcal{R}(\hat M, \hat m)$, then $\hat \pi$ becomes equivalent to $\hat \pi^{\mathsf{c}}$.
Since $w({\bar \pi}) \leq w(\hat M\mapsto M \mapsto \hat M \mapsto \bar{\bar \pi} \mapsto \hat m \mapsto m \mapsto \hat m)$, it trivially follows that $w(\hat \pi^{\mathsf{c}})\leq w(\hat \pi)$.
Moreover, note that $ w(\hat \pi)=w(\pi^{\mathsf{c}})$.
Hence, we obtain that $w(\hat \pi^{\mathsf{c}})\leq w(\pi^{\mathsf{c}})$.
Since $\pi^{\mathsf{c}}$ is optimal, the above inequality implies that $w(\hat \pi^{\mathsf{c}})=w(\pi^{\mathsf{c}})$ and that $\hat \pi^{\mathsf{c}}$ is optimal, too.
\end{proof}

The above result is now complemented with a useful characterization for optimal rides, which applies to the case when $m \leq M$ holds in
Theorem~\ref{thm:path}.

\begin{theorem}\label{thm:due}
Assume that there are two nodes $M,m\in V_{\mathcal{C}}\cup\{s_0,t_0\}$, with $s_0\leq M$, $m\leq t_0$ and $m \leq M$, such that the $(M, m)$-canonical ride $\pi^{\mathsf{c}}$ is an optimal ride.
Consider the  set
{\small
\begin{eqnarray*}\label{set:z}
\hat Z_m  & =&   \{ z\in \{s_0,t_0\}\cup V_\C \mid m \leq z \mbox{ and } s_0 \leq z \ \wedge \nexists (s,t)\in \mathcal{C}\mbox{ with } t < m \mbox{ and } z < s\}.
\end{eqnarray*}
}
It holds that $\hat Z \neq \emptyset$.
Moreover,  let $$\hat M_m=\min_{\hat z \in \hat Z_m} \hat z,$$
then $s_0\leq \hat M$, $m \leq \hat M_m$ and the $(\hat M_m,  m)$-canonical ride $\hat \pi^{\mathsf{c}}$ is optimal, too.
\end{theorem}
\begin{proof}
According to Definition~\ref{def:canonical},  since $m\leq M$, $\pi^{\mathsf{c}}$ has the form $\pi'\mapsto\pi''\mapsto\pi'''$ where:
$\pi'=s_0\mapsto M\mapsto \l(\R)\mapsto M$; $\pi''=M \mapsto \r(\mathcal{R})$; and  $\hat \pi'''=\r(\mathcal{R})\mapsto m \mapsto t_0$.
Note that there is no request $(s,t)\in \mathcal{C}$ such that  $t<m$ and $M < s$. Indeed, let us assume by the way of contradiction, that such
request exists. Note that, from the definition of $\pi^{\mathsf{c}}$,  there  is no pair of time steps $i$ and $i'$ such that $1\leq i\leq
i'\leq \len(\pi^{\mathsf{c}})$ with $M < \pi^{\mathsf{c}}_{i}$ and $\pi^{\mathsf{c}}_{i'}<m$. This implies that $\pi^{\mathsf{c}}$ does not
satisfies $(s,t)$, hence contradicting the feasibility of $\pi^{\mathsf{c}}$. The non existence of any request $(s,t)$ with $t<m$ and $M < s$,
implies that $M$ belongs to $\hat Z$, and hence $\hat Z \neq \emptyset$.

Let us prove now the next statements.  Note that $s_0\leq \hat M_m$ and $m \leq \hat M_m$  follow directly from the definition of  $\hat Z$. It
remains to show that $\hat \pi^{\mathsf{c}}$ is an optimal ride. According to  Definition~\ref{def:canonical},  since $m \leq \hat M_m$, $\hat
\pi^{\mathsf{c}}$ has the form  $\hat \pi'\mapsto \hat \pi''\mapsto \hat \pi'''$ where: $\pi'=s_0\mapsto \hat M_m\mapsto \l(\R)\mapsto \hat
M_m$; $\pi''=\hat M_m \mapsto \r(\mathcal{R})$; and  $\hat \pi'''=\r(\mathcal{R})\mapsto m \mapsto t_0$.
Let us show now  that $\hat \pi^{\mathsf{c}}$ is feasible.  Indeed, consider any request $(s,t)\in \mathcal{C}$.
In the case where $s\leq t$,  the request is satisfied by $\l(\R)\mapsto \r(\R)$, and hence by $\hat \pi^{\mathsf{c}}$.
Consider then the case where $t < s$.
Since $\hat M_m \in \hat Z$, it is not possible that $t< m$ and $\hat M_m< s$.
If $s\leq \hat M_m$, then $(s,t)$ is satisfied by $\hat \pi'$; if $m\leq t$ then $(s,t)$ is satisfied by $\hat \pi'''$.
So, in all the possible cases, $(s,t)$ is satisfied by  $\hat \pi^{\mathsf{c}}$, which implies that the canonical ride $\hat \pi^{\mathsf{c}}$ is a feasible ride.
In order to prove that $\hat \pi^{\mathsf{c}}$ is also optimal, we compare the cost of $\hat \pi^{\mathsf{c}}$ with the cost of the optimal ride $\pi^{\mathsf{c}}$.
Let us first notice that, since $\hat M_m$, by definition, is the smallest element in $\hat Z$ and $M$ belongs to $\hat Z$,  we get that $\hat M_m\leq M$ holds.
Consider the ride $\hat \pi=\hat \pi'\mapsto \ddot \pi''\mapsto \hat \pi'''$, where $\hat \pi'$ and $\hat \pi'''$ are the sub-rides defined for $\hat \pi^{\mathsf{c}}$, and
$$\ddot \pi''= \hat M_m \mapsto M \mapsto \hat M_m \mapsto \r(\mathcal{R}).$$
Note that,  if in $\ddot \pi''$ we replace  $\hat M_m \mapsto M \mapsto \hat M_m$ with $\hat M_m$, then $\hat \pi$ becomes equivalent to $\hat \pi^{\mathsf{c}}$.
 It trivially follows that $w(\hat \pi^{\mathsf{c}})\leq w(\hat \pi)$.
Moreover, note that $ w(\hat \pi)=w(\pi^{\mathsf{c}})$.
Hence, we obtain that $w(\hat \pi^{\mathsf{c}})\leq w(\pi^{\mathsf{c}})$.
Since $\pi^{\mathsf{c}}$ is optimal, the above inequality implies that $w(\hat \pi^{\mathsf{c}})=w(\pi^{\mathsf{c}})$ and that $\hat \pi^{\mathsf{c}}$ is optimal.
\end{proof}



\begin{algorithm}[th!]
\SetKwInput{KwData}{Input} \SetKwInput{KwResult}{Output}

\Indm \KwData{A ride-sharing scenario $\mathcal{R}=\tuple{G,(s_0,t_0),\mathcal{C}}$, where $G$ is a path and with
$\{s_0,t_0\} \cap \{ v\in V \mid \l(\mathcal{R})< v < \r(\mathcal{R})\}\neq \emptyset$;\\
\quad\quad Optionally, a Boolean value $\it symmetric$---set to \texttt{false}, if not provided;
}%
\KwResult{An optimal ride for $\mathcal{R}$;}

\Indp


  \tcc{PHASE I: implementation of Theorem~\ref{thm:uno}} Compute $\hat M$ and $\hat m$,  as defined in Theorem \ref{thm:uno}; \ \ \ // note
that $\hat M<\hat m$\label{step:comput1}

  $\pi^* \GETS$ any $(\hat M, \hat m)$-canonical ride; \ \ \ // use {\sc RideOnPath\_Outer} as a subroutine for $\R(\hat M,\hat
  m)$\label{step:comput1:ride}

  \tcc{PHASE II: implementation of Theorem~\ref{thm:due}}

  \For{each node $m \in V_{\mathcal{C}}\cup\{s_0,t_0\}$ with $m\leq t_0$}{
	Compute $\hat M_m$,  as defined in Theorem \ref{thm:due}; \ \ \ // note
that $\hat M_m \geq \hat m$\label{step:comput2}

    $\pi \GETS$ the  $(\hat M_m, m)$-canonical ride; \ \ \ // $s_0\mapsto \hat M_m \mapsto \l(\R)\mapsto \r(\R)\mapsto m\mapsto
    t_0$\label{step:comput2:ride}

    \If{$w(\pi) < w(\pi^*)$}{$\pi^* \GETS \pi$\;\label{step:if_phase2}}

  }

  \tcc{PHASE III: working on the symmetric scenario}

  \If{$\it symmetric$ is {\em \texttt{false}}}{

  $\pi^*_{\sym}  \GETS$ {\sc RideOnPath\_Inner}$(\sym(\R),\mathtt{true})$\;\label{step:symmetric}

  \If{$w(\pi^*_{\sym}) < w(\pi^*)$\label{step:compare}}{$\pi^* \GETS \sym(\pi^*_{\sym})$\;}\label{step:if_phase3}

  }

  \Return{$\pi^*$;}\label{step:return}

\caption{{\sc RideOnPath\_Inner}}\label{alg:main}
\end{algorithm}

In the light of  Theorem \ref{thm:path}, Theorem~\ref{thm:uno} and Theorem~\ref{thm:due}, consider then Algorithm~\ref{alg:main}, named {\sc
RideOnPath\_Inner}. It computes an optimal ride $\pi^*$ for the ``inner'' case, by
%
proceeding in three phases.

In Phase I, the algorithm computes the values $\hat M$ and $\hat m$ defined in Theorem~\ref{thm:uno} (step~\ref{step:comput1}), it builds a
$(\hat M, \hat m)$-canonical ride, and it assigns it  to $\pi^*$ (step~\ref{step:comput1:ride}). Note that, according to
Definition~\ref{def:canonical} and given that $\hat M<\hat m$, in order to build a $(\hat M, \hat m)$-canonical ride we need to compute an
optimal ride for $\R(\hat M, \hat m)$, which is a task that we can accomplish by exploiting {\sc RideOnPath\_Outer} as a subroutine---indeed,
note that $\R(\hat M, \hat m)$ fits the ``outer'' case.

In Phase II, the algorithm iterates over all possible values for $m$ in $V_{\mathcal{C}}\cup\{s_0,t_0\}$ with $m\leq t_0$. For each node $m$,
the value $\hat M_m$, defined in Theorem~\ref{thm:due}, is calculated (step \ref{step:comput2}). Then, the $(\hat M_m, m)$-canonical ride $\pi$
is built. In particular, since $\hat M_m\geq m$ holds, the ride $\pi$ is completely determined by Fact~\ref{fact:CR}. Eventually, if the cost
of $\pi$ is smaller than the cost of the current value of $\pi^*$, it updates $\pi^*$ to $\pi$ (step~\ref{step:if_phase2}).

Finally, Phase III is devoted to deal with the symmetric scenario $\sym(\R)$. The idea is that the first two phases are executed again on
$\sym(\R)$. Let $\pi^*_{\sym}$ be the result of this computation (step \ref{step:symmetric}). Then, we consider the symmetric ride
$\sym(\pi^*_{\sym})$, which is a ride for $\R$, and we compare its cost with the cost of the current value of $\pi^*$ (step
\ref{step:compare}).
As usual, we keep the ride with the associated minimum cost, which is eventually returned as output (step~\ref{step:return}).\\


The correctness of the method is proven below.

\begin{theorem}
Algorithm {\sc RideOnPath\_Inner} is correct.
\end{theorem}
\begin{proof}
Let us distinguish between two mutually exclusive cases:
\begin{enumerate}
\item[\emph{(1)}] $\R$ admits an optimal ride $\pi$ with $\leftIndex(\pi) < \rightIndex(\pi)$,
\item[\emph{(2)}] Every optimal ride $\pi$ for $\R$ is such that $\leftIndex(\pi) > \rightIndex(\pi)$.
\end{enumerate}

For \emph{(1)}, by combining  Theorem \ref{thm:path} with Theorem \ref{thm:uno} and Theorem \ref{thm:due}, we get that either
any $(\hat M, \hat m)$-canonical ride  is optimal, or there is a node $m\in   V_{\mathcal{C}}\cup\{s_0,t_0\}$ for which the $(\hat M_m,  m)$-canonical ride 
is optimal.
For \emph{(2)}, we notice that  $\sym(\R)$ admits an optimal ride that meets the condition of case \emph{(1)}. This implies that we can reduce case \emph{(2)} to
case \emph{(1)} by exploiting Fact \ref{fact:symmetric}.
We can conclude that an optimal ride for $\R$ is one with the smallest cost among any $(\hat M, \hat m)$-canonical ride and every  $(\hat M_m,
m)$-canonical ride, for every value of $m$ in $V_{\mathcal{C}}\cup\{s_0,t_0\}$, both for $\R$ and for $\sym(\R)$.

Note that {\sc RideOnPath\_Inner} exhaustively searches among all the possible candidate optimal rides listed above. Indeed, during Phase I,
the algorithm computes an $(\hat M, \hat m)$-canonical ride. During Phase II, the algorithm computes the best $(\hat M_m,  m)$-canonical ride,
for all possible values for $m$. Finally, during Phase III, the algorithm repeats the same computation for $\sym(\R)$. The algorithm returns
the ride  with the smallest cost among the ones which have been calculated. Hence the claim follows.
\end{proof}


\subsection{Implementation issues and running time}\label{sec:implementation}

In this section we analyze a concrete implementation and the corresponding running time of the algorithms we have proposed.
%
In fact, our goal is to prove the following theorem.

\begin{theorem}\label{thm:impl}
Let $\mathcal{R}=\tuple{G,(s_0,t_0),\mathcal{C}}$ be a ride-sharing scenario where $G=(V,E,w)$ is a path. Then, an optimal ride for $\R$
(together with its cost) can be computed in time $O(|\C| \log|\C|+|V|)$.
\end{theorem}

Note that checking whether an instance fits the  ``outer'' or the ``inner'' case is feasible in $O(|\C|)$. Then, we show that {\sc
RideOnPath\_Outer} and {\sc RideOnPath\_Inner} can be made to run in $O(|\C| \log |\C|+|V|)$.

\def\D{\mathcal{D}}

\subsubsection{{\sc RideOnPath\_Outer}}\label{sec:implementation:outer}

The running time of {\sc RideOnPath\_Outer} is essentially given by the running time of {\sc Normalize}. In particular, note that, in the case
where $s_0>t_0$, there is no need to materialize the symmetric scenario $\sym(\pi)$, since we can work on the original scenario by just
defining a function mapping each node $v\in V$ to its symmetric counterpart $\sym(v)=n-v+1$.

Concerning the implementation of {\sc Normalize}, we have first to build the set $\hat\C$ consisting of all requests $(s,t)$ with $t < s$ (cf.
step~\ref{alg:step1}).
Actually, we propose to sort these requests in order of starting node and, accordingly, we shall assume that $\hat \C = \{(s_1,t_1),
(s_2,t_2),\ldots, (s_{|\hat \C|},t_{|\hat \C|})\}$ holds with $s_i \leq s_j$ whenever $i < j$.
Similarly, we sort the nodes in  $V_{\C} \cup \{s_0, t_0\}$, and hence we assume that $V_{\C} \cup \{s_0, t_0\} = \{w_1, w_2,\ldots, w_r\}$
holds with $w_i \leq w_j$ whenever $i < j$. Moreover, for each node $w_i  \in V_{\C} \cup \{s_0, t_0\}$, we define the set $F(w_i) = \{j \mid
(s_j, t_j)\in \hat\C \wedge (w_i = s_j  \mbox{ or } w_i=t_j)\}$, maintained as linked list. And, finally, for each element $j$ in $F(w_i)$ we
keep a label $l_{ij}\ \in \{{\sf s}, {\sf t}\}$ denoting whether $w_i$ is a starting ($\sf s$) or a terminating ($\sf t$) node of request $j$.
Note that step~\ref{alg:step1} plus the construction of such data structures are clearly feasible in $O(|\C|\log |\C|)$.

Consider now the steps \ref{alg:step3:1}-\ref{alg:step3:2} and \ref{alg:step2:1}-\ref{alg:step2:2}. For any set of requests $\D$ on $G$ and
every node $v\in V_{\D}$, let $T^{1}_{v}(\D) = \{(s,t)\in \D \mid t = v < s\}$, $T^{2}_{v}(\D) = \{(s,t)\in \D \mid t < v < s\}$, and
$T^{3}_{v}(\D) = \{(s,t)\in \D \mid t < v = s\}$. Moreover, let $L(\D) = \{v\in V_{\D} \mid T^{1}_{v}(\D) \neq \emptyset \mbox{ and
}T^{2}_{v}(\D) = T^{3}_{v}(\D) = \emptyset\}$, and $R(\D) = \{v\in V_{\D} \mid  T^{1}_{v}(\D) = T^{2}_{v}(\D) = \emptyset \mbox{ and }
T^{3}_{v}(\D) \neq \emptyset\}$.
We use the following technical ingredient.

\begin{claim}\label{claim:norm}
Let $\C^* = \{(s^*_1, t^*_1), (s^*_2, t^*_2),\ldots, (s^*_h, t^*_h)\}$ be the output of {\sc Normalize}. Then, the following properties hold:
\begin{enumerate}
\item[(1)] $L(\hat\C) = \{t^*_1, t^*_2,\ldots, t^*_h\}$ and $R(\hat\C) = \{s^*_1, s^*_2,\ldots, s^*_h\}$;\label{claim:norm:1}
\item[(2)] $s^*_i = \min_{v\in R_i} v$, where $R_i = \{v\in R(\hat\C) \mid v \geq t^*_i\}$, for every $1 \leq i \leq
    h$.\label{claim:norm:2}
\end{enumerate}
\end{claim}
\begin{proof}
For \emph{(1)}. It is immediate that $L(\C^*) = \{t^*_1, t^*_2,\ldots, t^*_h\}$. Hence, our proof consists in showing that $L(\hat\C) =
L(\C^*)$.
Let $\hat\C = \D_0, \D_1, \ldots, \D_p$ be the sequences of requests produced during the execution of steps~\ref{alg:step3:1} and
\ref{alg:step3:2}, i.e., for every $0\leq i\leq p-1$, $\D_{i+1}$ is the set of requests obtained from $\D_{i}$ after performing one iteration
of the while loop. We show by induction that $L(\D_i) = L(\D_0)$, for every $0\leq i \leq p$. The base case trivially holds. Let us suppose
that, for a given $0\leq k \leq p$, $L(\D_k) = L(\D_0)$ holds. We must show that $L(\D_{k+1}) = L(\D_0)$ holds, too.
Let $(s,t), (s',t')$ be two requests in $\D_{k}$ such that $t<s$, $t'<s'$ and $t' \leq t \leq s' \leq s$; and, let $\D_{k+1} = (\D_{k}
\setminus \{(s,t), (s',t')\}) \cup \{(s,t')\}$.
Note that every node $v$ such that $v < t'$ or $s < v$ belongs to $L(\D_{k+1})$ if, and only if, it belongs also to $L(\D_k)$; every node $v$
such that $t' < v \leq s$ belongs neither to $L(\D_k)$ nor to $L(\D_{k+1})$; finally, $t'$ belongs to $L(\D_{k+1})$ if, and only if, it belongs
to $L(\D_k)$. We can conclude that $L(\D_p)=L(\D_0)$.
Now, let $\D_p, \D_{p+1},\ldots,\D_q = \C^*$  be the sequences of requests produced during the execution of steps~\ref{alg:step2:1} and
\ref{alg:step2:2}, i.e., for every $p\leq i\leq q-1$, $\D_{i+1}$ is the set of requests obtained from $\D_{i}$ after performing one iteration
of the while loop. Again, we show by induction that $L(\D_i) = L(\D_p)$, for every $p\leq i \leq q$.
The base case trivially holds. Let us suppose that for a given $p\leq k \leq q$, $L(\D_k) = L(\D_p)$ holds. We must show that $L(\D_{k+1}) =
L(\D_p)$ holds, too. Let $(s,t), (s',t')$ be two requests in $\D_{k}$ such that $t'\leq t < s \leq s'$;  and let $\D_{k+1} = \D_{k} \setminus
\{(s,t)\}$. Note that every node $v$ such that $v < t'$ or $s' < v$ belongs to $L(\D_{k+1})$ if, and only if, it belongs also to $L(\D_k)$;
every node $v$ such that $t' < v \leq s'$  belongs neither to $L(\D_k)$ nor to $L(\D_{k+1})$; finally, $t'$ belongs to $L(\D_{k+1})$ if, and
only if, it belongs also to $L(\D_k)$. We can finally conclude that $L(\D_q)=L(\D_p)=L(\D_0)$.
Similar arguments can be used to show that  $R(\hat\C) = \{s^*_1, s^*_2,\ldots, s^*_h\}$

\smallskip

For \emph{(2)}. By the way of contradiction, let us assume that the claim is not true. Let $j$ be the smallest index such that $s^*_j > s^*_k$,
where $s^*_k = \min_{v\in R_j} v$. This implies that $s^*_j>s^*_k\geq t^*_j$, which is impossible since $\C^*$ is in normal form (cf.
Lemma~\ref{claim:normalize}).
\end{proof}

According to Claim~\ref{claim:norm}, in order to determine the set of requests produced as output by {\sc Normalize}, we can iterate through
the nodes in $V_{\C} \cup \{s_0, t_0\}$ in order of increasing index, starting from $w_1$. We maintain three sets of indexes of requests in
$\hat \C$, namely $S_1$, $S_2$ and $S_3$. Moreover, we maintain two sets of nodes $Q_L$ and $Q_R$. Initially, $S_1=S_2=S_3=\emptyset$ and $Q_L
= Q_R = \emptyset$. At the beginning of $k$-th iteration, we set $S_3$ to the empty set, and we move all the elements in $S_1$ to $S_2$. Then,
we move from $S_2$ to $S_3$ every $j \in F(w_k)$ with $l_{kj} = {\sf s}$, and we add to $S_1$ every $j \in F(w_k)$ with $l_{kj} = {\sf t}$.
Thus, at the end of the iteration, $S_1$, $S_2$ and $S_3$ contain all the elements in $T^1_{w_k}$,  $T^2_{w_k}$ and $T^3_{w_k}$, respectively.
Hence, at the end of the $k$-th iteration, if $S_1 \neq \emptyset$, $S_2 = \emptyset$ and $S_3 = \emptyset$, then we add $w_k$ to $Q_L$;
otherwise, if $S_1 = \emptyset$, $S_2 = \emptyset$ and $S_3 \neq \emptyset$, then we add $w_k$ to $Q_R$. We continue in this fashion until we
run out of nodes. Because of Claim~\ref{claim:norm}, after we iterate through all nodes, $Q_L$ and $Q_R$ consist of all nodes in $L(\hat\C)$
and $R(\hat\C)$, respectively.
Eventually, in order to build the normalized scenario, we can just pair, by Claim~\ref{claim:norm}, every node $t$ in $L(\hat C)$ with the
smallest node $s$ in $R(\hat \C)$ larger than $t$.

%
%

Note that every request in $\hat \C$ is added and removed exactly once from each of the three sets $S_1$, $S_2$ and $S_3$. Moreover, each node
in $V_{\hat\C}$ is added and removed at most once from either $Q_L$ or $Q_R$. Hence, the time taken by the procedure is at most $O(|\hat\C|)$
times the maximum cost for performing each operation. If the set $S_2$ is maintained as a binary min-heap,  where the key of each request is
its starting node,  removing an element from $S_2$ with label {\sf s} corresponds to extract the element with smallest key,  and both the
insertion and the removal from $S_2$ can be made to run in time $O(\log |\hat \C|)$. On the other side, since each removal from $S_1$ and $S_3$
is performed without making any distinction among elements, we can easily keep constant the cost of each insertion and removal from $S_2$, by
maintaining both $S_1$ and $S_3$ as a linked list. Finally, if both $Q_L$ and $Q_R$ are maintained as a binary min-heap, where the key of each
node is the node itself, removing the smallest node from the set corresponds to extract the element with smallest key, and both the insertion
and the removal can be made to run in time $O(\log |\hat \C|)$. Summarizing, every insertion and removal takes at most $O(\log |\hat \C|)$.
Thus, our implementation of {\sc RideOnPath\_Outer} takes total time $O(|\hat \C|\log |\hat \C|)$. Since $\hat C\subseteq \C$, the algorithm
takes $O(|\C|\log |\C|)$.

Actually, note that the algorithm produces a result that is given in the form $s_0\mapsto x_1\mapsto ... \mapsto x_m\mapsto t_0$, where
$x_1,...,x_m$ are nodes of the graph and $m=O(|\C|)$ holds. Basically, this is a succinct representation consisting of listing (at least) all
the nodes where the current direction of traversing the path has to be reverted. Of course, to explicitly build the ride and compute the
associated cost takes an extra $O(|V|)$ time.

\subsubsection{{\sc RideOnPath\_Inner}}\label{sec:implementation:inner}

Let us now move to analyze {\sc RideOnPath\_Inner} and let us focus on Phase I and Phase II (again, working on the symmetric scenario is
immediate).
Phase I starts with the computation of $\hat M$ and $\hat m$. Let us discuss the procedure to compute $\hat M$. According to Theorem
\ref{thm:uno}, $\hat M$ is defined as the smallest node in $\hat X$. Hence, in order to compute $\hat M$, we iterate through the nodes in
$V_{\C} \cup \{s_0, t_0\}$ in order of increasing index, until we find a node in $\hat X$. There is a easy method to determine if a node
belongs to $\hat X$. For every node $w_i \in V_{\C} \cup \{s_0, t_0\}$, let $P_{w_i} = \{(s,t)\in \C \mid t \leq w_i < s\}$. It is easy to see
that a node $w_i \in V_{\C} \cup \{s_0, t_0\}$ belongs to $\hat X$ if, and only if,
%
$w_i \geq s_0$ and $P_{w_i} = \emptyset$. Note that $P_{w_i} \subseteq \hat\C$, where $\hat C$ is the set of requests built in Section
\ref{sec:implementation:outer}. Hence, we can write $P_{w_i} = \{(s,t)\in \hat\C \mid t \leq w_i < s\}$ and in the following we use the same
datastructures discussed for the implementation of {\sc RideOnPath\_Outer}.

More specifically, the algorithm works as follows. We iterate through the nodes in $V_{\C} \cup \{s_0, t_0\}$ in order of increasing index,
starting from $w_1$. Throughout the iteration, we maintain a set $S$ of indexes of requests in $\hat \C$. Initially $S=\emptyset$; during the
$k$-th iteration, we add to $S$ every $j \in F(w_k)$ with $l_{kj} = {\sf t}$, and we remove from $S$ every $j \in F(w_k)$ with $l_{kj} = {\sf
s}$.
Note that, at the end of the iteration,
$S$ contains all the elements in $P_{w_k}$, so that if 
$w_k \geq s_0$ and $S = \emptyset$, then we terminate by concluding that $w_k$ is the smallest element in $\hat X$. Given the existence of
$\hat M$, such procedure always terminates.
For the complexity analysis, observe that every request in $\hat \C$ is added and removed from $S$ exactly once. Hence, the time taken by the
procedure is at most $O(|\hat\C|)$ times the maximum cost for performing each operation. If the set $S$ is maintained as a binary min-heap,
where the key of each request is its starting node, removing an element from $S$ with label {\sf s} corresponds to extract the element with
smallest key, and both the insertion and the removal can be made to run in time $O(\log |\hat \C|)$. A similar approach can be used to compute
$\hat m$.
Thus, Phase I takes total time $O(|\hat \C|\log |\hat \C|)$, hence $O(|\C| \log |\C|)$, to define the pair $\hat M,\hat m$. A canonical ride
with its associated cost can be then computed in $O(|\C| \log |\C|+|V|)$, since the dominant operation is the invocation of the algorithm for
the outer case (cf. Section \ref{sec:implementation:outer}).

\smallskip

Phase II starts with the computation of $\hat M_{w_i}$, for every node $w_i$ in $V_{\mathcal{C}}\cup\{s_0,t_0\}$ with $w_i \leq t_0$. For an
efficient computation, we use the following technical claim.

\begin{claim}\label{claim:MM}
For every node $m \in V_{\mathcal{C}}\cup\{s_0,t_0\}$ with $m \leq t_0$, let $\hat M_m$ be the node as defined in Theorem \ref{thm:due}.
Consider the set $Q_m = \{(s', t') \in \C  \mid  t' < m < s'\}$, and let
$$
u_m = \begin{cases}
\max\{m,\; s_0\} & \mbox{if }\;\;  Q_m = \emptyset,\\
\max\{s_0, \;\max_{(s', t')\in Q_m} s'\} & \mbox{otherwise.}
\end{cases}
$$
Then
$\hat M_m = u_m$.
\end{claim}
\begin{proof}
We prove the claim by showing that $u_m$ belongs to $\hat Z_m$, and every other node $v\in V_{\hat\C}$ such that $v < u_m$ does not belong to
$\hat Z_m$. This implies that $u_m$ is the smallest element in $\hat Z_m$, hence it coincides with $\hat M_m$. Let us recall that $\hat Z_m$ is
the set of all nodes $z$ in  $ \{s_0,t_0\}\cup V_\C$ such that (1) $m \leq z \mbox{ and } s_0 \leq z$; and (2) $\nexists (s,t)\in
\mathcal{C}\mbox{ with } t < m \mbox{ and } z < s$.

Assume that $Q_m = \emptyset$. In this case $u_m = \max\{m, s_0\}$, and every node in  $ \{s_0,t_0\}\cup V_\C$ satisfies condition (2). It is
easy to verify that $u_m$ always satisfies condition (1) and every node strictly smaller than $u_m$ does not belong to $\hat Z_m$.
Assume now that $Q_m \neq \emptyset$.  In this case $u_m = \max\{s_0, \;\max_{(s', t')\in Q_m} s'\}$. Also in this case, it is easy to verify
that $u_m$ always satisfies condition (1). By the way of contradiction, let us assume that condition (2) is not satisfied, that is, there
exists a request $(s,t)$ with $t < m$ and $u_m < s$. Note that such request necessarily belongs to $Q_m$, which implies that $u_m \geq s$, a
contradiction. Finally, let us prove that $u_m$ is the smallest value in $\hat Z_m$ by showing that any other node strictly smaller than $u_m$
violates one of the two conditions. If $s_0 \geq (\max_{(s', t')\in Q_m} s')$ then $u_m = s_0$; in this  case every node strictly smaller than
$s_0$ does not satisfies condition (1). Instead,  if $s_0 < (\max_{(s', t')\in Q_m} s')$ then $u_m = (\max_{(s', t')\in Q_m} s')$. In this
latter case, let $(s,t)$ be the request in $Q_m$ with the largest starting node, i.e.,  $t < m$ and $u_m = s$. If we take any other node $v$
strictly smaller than $u_m$, than we get $t < m$ and $v < s = u_m$, hence violating (2).
\end{proof}

According to Claim \ref{claim:MM}, for every node $w_i \in V_{\C} \cup \{s_0, t_0\}$,  $\hat M_{w_i}$ is defined as the maximum between $w_i$
and $s_0$, if $Q_{w_i}$ is not empty, or the maximum between $s_0$ and $\max_{(s', t')\in Q_{w_i}} s'$, otherwise.
So, the dominant operation is the computation of $Q_{w_i}$. To this end, for every $w_i \in V_{\C} \cup \{s_0, t_0\}$, we iterate through the
nodes in $ V_{\C} \cup \{s_0, t_0\}$ in order of increasing index. Note that $Q_{w_i} \subseteq \hat\C$, hence equivalently we can write
$Q_{w_i} = \{(s', t') \in \hat\C  \mid  t' < w_i < s'\}$; this implies that, in order to compute $Q_{w_i}$, we need of only the requests in
$\hat\C$ and we can use the usual data structures.

More specifically, we iterate through the nodes in $V_{\C} \cup \{s_0, t_0\}$ in order of increasing index, starting from $w_1$. Initially, we
define a set $S=\emptyset$. During the $k$-th iteration, we remove from $S$ every $j \in F(w_k)$ with $l_{kj} = {\sf s}$, and if $k \geq 2$ we
add to $S$ every $j \in F(w_{k-1})$ with $l_{(k-1)j} = {\sf t}$.
Note that, at the end of the iteration, $S$ contains all the elements in $Q_{w_k}$. Thus, if $S=\emptyset$, then we set $M_{w_k}$ to $\max\{m,
s_0\}$, otherwise we set $M_{w_k}$ to $\max\{s_0, \;\max_{(s', t')\in S}\; s'\}$. In the latter case, we need to calculate $\max_{(s', t')\in
S}\; s'$, i.e., to search in $S$ for the request with the largest starting node. We continue in this fashion until we run out of nodes.
For the complexity analysis, observe that every request in $\hat \C$ is added and removed from $S$ exactly once. Moreover, at the end of each
iteration, we need to search in $S$ for the request with the largest starting node, in order to calculate $\max_{(s', t')\in S}\; s'$. Hence,
the time taken by the procedure is at most $O(|\hat\C|)$  times the maximum cost for performing each operation. If the set $S$ is maintained as
a binary min-max-heap,  where the key of each request is its starting node,  removing an element from $S$ with label {\sf s} corresponds to
extract the element with smallest key, hence both the insertion and the removal can be made to run in time $O(\log |\hat \C|)$; moreover,
calculating $\max_{(s', t')\in S}\; s'$ corresponds to search for the element with largest key,  which takes only constant time. Thus, the
computation of $\hat M_{w_i}$, for every node $w_i \in V_{\C} \cup \{s_0, t_0\}$, takes a total time $O(|\hat \C|\log |\hat \C|)$, hence
$O(|\C|\log |\C|)$.

\smallskip

Now, note that the computation of the $(\hat M_m,m)$-canonical ride takes constant time, since by Fact~\ref{fact:CR}, we know that this ride
has the form $s_0\mapsto \hat M_m \mapsto \l(\R)\mapsto \r(\R)\mapsto m\mapsto t_0$. Then, the remaining operation in Phase II is the
comparison between the cost of the given best ride and cost of the current ride. We have already seen that the computation of the cost of rides
built in Phase I can be accommodated in the overall $O(|\C|\log |\C|+|V|)$ cost. Now, we claim that the computation of the cost of the $(\hat
M_m,m)$-canonical ride takes constant time, provided a suitable pre-processing. Indeed, observe that the $(\hat M_m,m)$-canonical ride is
succinctly represented by a constant number of nodes. The idea is then to associate each node $x\in V$ with the value $cw(x)=\sum_{i=2}^x
w(\{i,i+1\})$, which is overall feasible in $O(|V|)$. Then, the cost for a rides moving from a node $x$ to a node $y$, along the unique path as
defined in the notion of canonical ride, is just given by the value $|cw(y)-cw(x)|$. Therefore, with a constant overhead, the cost of the 
$(\hat M_m,m)$-canonical ride can be computed. Putting it all together, Phase II can be implemented in $O(|\C|\log |\C|+|V|)$, too.

\section{Optimal Rides on Cycles}\label{sec:cycles}
\def\mod{\mathtt{\;mod\;}}

In this section, we consider scenarios $\mathcal{R}=\tuple{G,(s_0,t_0),\mathcal{C}}$ such that the underlying graph $G=(V, E, w)$, with
$V=\{1,\ldots,n\}$, is a \emph{cycle}. Formally, for each node $v\in V\setminus \{n\}$, the edge $\{v, v+1\}$ is in $E$; moreover, the edge
$\{n,1\}$ is in $E$; and no further edge is in $E$. Without loss of generality, we assume $s_0=1$.

\subsection{From Cycles to Paths}

The solution approach we shall propose is to reuse the methods we have already developed to deal with scenarios over paths. In this section, we
define the key technical ingredients, and based on them an algorithm will be subsequently illustrated.


Let $\pi$ be a ride on $\mathcal{R}$, and let us associate each of its time steps $i$ with a ``virtual'' node
$\tau_\pi(i)=\pi_i+(\ell_\pi(i)-\min_{j\in\{1,\dots,\len(\pi)\}} \ell_\pi(j))\cdot n$, where $\ell_\pi(1)=0$  and where, for each
$i\in\{2,\dots,\len(\pi)\}$, $\ell_\pi(i)$ is an integer defined as follows:
$$
\ell_\pi(i)=\left\{\begin{array}{ll}
\ell_\pi(i-1)+1 & \mbox{if $\pi_{i-1}=n$ and $\pi_i=1$}\\
\ell_\pi(i-1)-1 & \mbox{if $\pi_{i-1}=1$ and $\pi_i=n$}\\
\ell_\pi(i-1) & \mbox{otherwise}
\end{array}\right.
$$

Intuitively, the function $\tau_\pi$ keeps track of the number of times in which the cycle is completely traversed by the ride, either
clockwise or anti clockwise. Note that $\tau_\pi(i) \mod n = \pi_i$.



Let $\cw(\pi)$ (resp., $\acw(\pi)$) be the maximum (resp., minimum) value of $\tau_\pi(i)$ over all time steps $i\in\{1,\dots,\len(\pi)\}$.
Let $\cwIndex(\pi)$ (resp., $\acwIndex(\pi)$) be the minimum time step $i\in\{1,\dots,\len(\pi)\}$ such that $\tau_\pi(i)=\acw(\pi)$ (resp.,
$\tau_\pi(i)=\cw(\pi)$). 
Note that $1\leq \acw(\pi)\leq n$ always hold, by definition of $\tau_\pi$. In fact, over optimal rides, useful characterizations and bounds
can be derived for both $\acw(\pi)$ and $\cw(\pi)$.

\begin{lemma}\label{lem:boundCW}
An optimal ride $\pi$ exists with $\cw(\pi)\leq 3n$ and $\{\cw(\pi)\ {\tt mod}\ n,\acw(\pi)\ {\tt mod}\ n\}\subseteq V_\mathcal{C} \cup
\{s_0,t_0\}$.
\end{lemma}
\begin{proof}
Assume that $\pi$ is an optimal ride for $\R$.
Assume that $\cw(\pi)\ {\tt mod}\ n$ (resp., $\acw(\pi)\ {\tt mod}\ n$) is not contained in $V_\mathcal{C} \cup \{s_0,t_0\}$.
Then, let us build a ride $\hat \pi$ from $\pi$ by removing all time steps $i$ such that $\tau_i(\pi)=\cw(\pi)$ (resp.,
$\tau_i(\pi)=\acw(\pi)$).
By definition of $\cw$ (resp. $\acw$), $\hat \pi$ is a feasible ride and $w(\hat \pi)\leq w(\pi)$. Therefore, $\hat \pi$ is an optimal ride,
too. Now, either $\hat \pi$ satisfies the desired condition, or the process can be iterated till a ride $\pi^*$ is obtained such that
$\{\cw(\pi^*)\ {\tt mod}\ n,\acw(\pi^*)\ {\tt mod}\ n\}\subseteq V_\mathcal{C} \cup \{s_0,t_0\}$.

Therefore, let us assume, w.l.o.g., that $\pi$ is an optimal ride with $\{\cw(\pi)\ {\tt mod}\ n,\acw(\pi)\ {\tt mod}\ n\}\subseteq
V_\mathcal{C} \cup \{s_0,t_0\}$. Consider the case where $\acwIndex(\pi)\leq \cwIndex(\pi)$---in fact, a similar argument applies when
$\acwIndex(\pi)>\cwIndex(\pi)$.
Assume, for the sake of contradiction, that $\cw(\pi)> 3n$. Since $\acw(\pi)\leq n$, this means that $\cw(\pi)-\acw(\pi)>2n$, and
hence, $\cwIndex(\pi)- \acwIndex(\pi)>2n$ holds, too. Let $i$ be the maximum time step such that $i\leq \cwIndex(\pi)$ and
$\pi_i=\pi_{\acwIndex(\pi)}$. Moreover, let $i'$ and $i''$ be two time steps with $i<i'<i''$ such that $\pi_i=\pi_{i'}=\pi_{i''}$. In
particular, let $i''$ be the maximum time step such that $\pi_i=\pi_{i'}=\pi_{i''}$. Given the above observations, $i'$ and $i''$ are well
defined. Indeed, starting from the time step $i$, $\pi$ must transverse clockwise the cycle twice. Furthermore, for the same reason, the
following ride
$$
\pi'=\pi[1,\acwIndex(\pi)],(\pi_i+1){\tt mod}\ n,\dots,(\pi_i+2n-1){\tt mod}\ n, \pi[i'',\len(\pi)].
$$

\noindent is such that $\pi'\preceq \pi$. In particular, note that $\pi'$ transverses the cycles twice too, and we have $\cw(\pi')\leq 3n$. In order to conclude the proof, note that $\cw(\pi)\ {\tt mod}\ n=\cw(\pi')\ {\tt mod}\ n$ and $\acw(\pi)\ {\tt mod}\ n=\acw(\pi')\ {\tt
mod}\ n$, and hence $\{\cw(\pi')\ {\tt mod}\ n,\acw(\pi')\ {\tt mod}\ n\}\subseteq V_\mathcal{C} \cup \{s_0,t_0\}$.
\end{proof}

Now, consider the path $G^\circ=(V^\circ,E^\circ,w^\circ)$, where $V^\circ=\{1,\dots,3n\}$ and where $w^\circ$ is the function such that
$w^\circ(\{v,v+1\})=w(\{v\ {\tt mod}\ n,(v+1){\tt mod}\ n\})$.

For each pair of nodes $\alpha,\beta\in V^\circ$ with $\alpha\leq \beta$, let us define $V_{\alpha,\beta}^\circ$ as the set of nodes $v\in
\{\alpha,\dots,\beta\}$ for which no other distinct node $v'\in \{\alpha,\dots,\beta\}$ exists such that $v\mod n=v'\mod n$.
Note that if $\beta< \alpha+n$, then $V_{\alpha,\beta}^\circ=\{\alpha,\dots,\beta\}$; if $\beta\geq \alpha+2n -1$, then
$V_{\alpha,\beta}^\circ=\emptyset$; if $\alpha+n \leq \beta<\alpha+2n-1$, then $V_{\alpha,\beta}^\circ=\{\beta-n+1,\dots,\alpha+n-1\}$.

Moreover, define $\mathcal{C}_{\alpha,\beta}^\circ=\{ (v_s,v_t) \mid (v_s\ {\tt mod}\ n,v_t\ {\tt mod}\ n)\in \mathcal{C}, v_s\in
V_{\alpha,\beta}^\circ, v_t\in V_{\alpha,\beta}^\circ \}.$

\begin{theorem}\label{thm:cp}
Let $\pi$ be a feasible ride for $\mathcal{R}$ with $\cw(\pi)\leq 3n$ and such that $\acwIndex(\pi)\leq \cwIndex(\pi)$ (resp.,
$\acwIndex(\pi)> \cwIndex(\pi)$).
Let $\alpha=\acw(\pi)$ and $\beta=\cw(\pi)$, and let $(s^\circ,t^\circ)=(\alpha,\beta)$ (resp., $(s^\circ,t^\circ)=(\beta,\alpha)$).
Then, the ride $$\tau_\pi(1),...,\tau_\pi(\len(\pi))$$ is feasible for
$\tuple{G^\circ,(\tau_\pi(1),\tau_\pi(\len(\pi))),\mathcal{C}_{\alpha,\beta}^\circ\cup\{(s^\circ,t^\circ)\}}$.
\end{theorem}
\begin{proof}
Let $\Upsilon=\tau_\pi(1),...,\tau_\pi(\len(\pi))$.
Note first that each node $v\in\nodes(\Upsilon)$ belongs to $V^\circ$, because $\cw(\pi)\leq 3n$.
Therefore, we have to show that $\Upsilon$ satisfies every request in $\mathcal{C}_{\alpha,\beta}^\circ$. In fact, $\Upsilon$ clearly satisfies
$(s^\circ,t^\circ)$.
Consider then any request $(v_s,v_t)\in \mathcal{C}_{\alpha,\beta}^\circ$ such that $(v_s\ {\tt mod}\ n,v_t\ {\tt mod}\ n)$ is a request in
$\mathcal{C}$ with $v_s\in V_{\alpha,\beta}^\circ$ and $v_t\in V_{\alpha,\beta}^\circ$. Since $\pi$ is feasible for $\R$, there are two time
steps $i$ and $j$ such that $i\leq j$, $\pi_i=v_s\ {\tt mod}\ n$ and $\pi_{j}=v_t\ {\tt mod}\ n$.
Actually, by definition of $\alpha$ and $\beta$, since $v_s\in V_{\alpha,\beta}^\circ$ (resp., $v_t\in V_{\alpha,\beta}^\circ$), there is no
different time step $i'$ (resp., $j'$) such that $\pi_{i'}=v_s\ {\tt mod}\ n$ (resp., $\pi_{j'}=v_t\ {\tt mod}\ n$). Hence, we have that
$\tau_\pi(i)=v_s$ and $\tau_\pi(j)=v_t$; in fact, $\tau_\pi$ restricted on $V_{\alpha,\beta}^\circ$ is a bijection. So, $\Upsilon$ satisfies
$(v_s,v_t)$.
\end{proof}

Intuitively, the result tells us that feasible rides for $\R$ are mapped into feasible rides for a suitable defined scenario over a path.
Below, we show that the converse also holds, under certain technical conditions.

\begin{theorem}\label{thm:cycle}
Consider the following setting:
\begin{itemize}
  \item[(i)] $\alpha,\beta\in V^\circ$ is a pair of nodes such that $\{\alpha\ {\tt mod}\ n,\beta\ {\tt mod}\ n\}\subseteq V_\mathcal{C}
      \cup \{s_0,t_0\}$, $1\leq \alpha,\beta\leq 3n$, and such that, for each $x\in V_\mathcal{C}\cup \{s_0,t_0\}$, there is a node $v_x\in
      V^\circ$ with $\alpha\leq v_x\leq \beta$ and $x=v_x \ {\tt mod}\ n$.

  \item[(ii)] $v_{s_0},v_{t_0}\in V^\circ$ is a pair of nodes such that $\alpha\leq v_{s_0}\leq \beta$, $\alpha \leq v_{t_0}\leq \beta$,
      $v_{s_0}\ {\tt mod}\ n=s_0$, and $v_{t_0}\ {\tt mod}\ n=t_0$.

  \item[(iii)] $(s^\circ,t^\circ)$ is a request such that $(s^\circ,t^\circ)\in\{(\alpha,\beta),(\beta,\alpha)\}$.
\end{itemize}

Let $\pi^\circ$ be a feasible ride for $\tuple{G^\circ,(v_{s_0},v_{t_0}),\mathcal{C}_{\alpha,\beta}^\circ\cup\{(s^\circ,t^\circ)\}}$.
Then,     $$ \pi^\circ_1\ {\tt mod}\ n, \dots,\pi^\circ_{\len(\pi^\circ)}\ {\tt mod}\ n $$ is a feasible ride for $\R$.
\end{theorem}
\begin{proof}
Let $\pi^\circ$ be a feasible ride for $\tuple{G^\circ,(v_{s_0},v_{t_0}),\mathcal{C}_{\alpha,\beta}^\circ\cup\{(s^\circ,t^\circ)\}}$, and let
$\Lambda$ be the ride such that:
$$\Lambda=\pi^\circ_1\ {\tt mod}\ n, \dots,\pi^\circ_{\len(\pi^\circ)}\ {\tt mod}\ n.$$

Note that $\pi^\circ_1=v_{s_0}$ and $\pi^\circ_{\len(\pi^\circ)}=v_{t_0}$. Because of \emph{(ii)}, $\Lambda_1=s_0$ and
$\Lambda_{\len(\Lambda)}=t_0$. Therefore, in order to show that $\Lambda$ is feasible for $\R$, we have to show that it satisfies each request
in $\mathcal{C}$.
Let $(s,t)$ be in $\mathcal{C}$. We distinguish two cases.

First, assume there is a pair $v_s,v_t$ of nodes in $V_{\alpha,\beta}^\circ$ such that $s=v_s\ {\tt mod}\ n$ and $t=v_t\ {\tt mod}\ n$. Then,
      $(v_s,v_t)$ is in $\mathcal{C}_{\alpha,\beta}^\circ$. By the feasibility of $\pi^\circ$, it follows that there are two time steps $i$ and
      $j$ with $i\leq j$ such that $\pi^\circ_i=v_s$ and $\pi^\circ_{j}=v_t$. Hence, $\Lambda_i=s$ and $\Lambda_{j}=t$, implying that $\Lambda$
      satisfies $(s,t)$, too.

Second, assume that $V_{\alpha,\beta}^\circ$ contains no node $v_s$ such that $s=v_s\ {\tt mod}\ n$; in fact, the case where
      $V_{\alpha,\beta}^\circ$ contains no node $v_t$ such that $t=v_t\ {\tt mod}\ n$ can be addressed with the same line of reasoning. Recall
      that, because of \emph{(i)}, for each $x\in V_\mathcal{C}\cup \{s_0,t_0\}$, there is a node $v_x\in V^\circ$ with $\alpha\leq v_x\leq
      \beta$ and $x=v_x \ {\tt mod}\ n$.
      Therefore, we conclude that there are two nodes $v_s<v_{s}'$ such that $\alpha\leq v_s$, $v_{s}'\leq \beta$, $s=v_s\ {\tt mod}\ n=v_s'\
      {\tt mod}\ n$. In this case, there must be a node $v_t$ such that $v_s\leq v_t\leq v_{s}'$ and $t=v_t\ {\tt mod}\ n$. Since $\pi^\circ$
      satisfies $(s^\circ,t^\circ)$ because of \emph{(iii)}, there is a pair of time steps $i$ and $j$ with $i\leq j$ and such that
      $\pi_i=s^\circ$ and $\pi_j=t^\circ$. Assume $(s^\circ,t^\circ)=(\alpha,\beta)$. Then, there is a pair of time instants $i^*,j^*$ such
      that $i\leq i^*\leq
       j^*\leq j$ and $\pi_{i^*}=v_s$ and $\pi_{j^*}=v_t$. Therefore, $\Lambda_{i^*}=s$, $\Lambda_{j^*}=t$, and thus $\Lambda$ satisfies
      $(s,t)$. To conclude, consider the case where $(s^\circ,t^\circ)=(\beta,\alpha)$. In this case, there is a pair of time instants
      $i^*,j^*$ such that $i\leq i^*\leq j^*\leq j$ and $\pi_{i^*}=v_s'$ and $\pi_{j^*}=v_t$. In fact, we still have $\Lambda_{i^*}=s$,
      $\Lambda_{j^*}=t$, and thus $\Lambda$ again satisfies $(s,t)$.
\end{proof}

\subsection{Putting It All Together}

Armed with the above technical ingredients, we can now illustrate Algorithm~\ref{alg:cycle}, named {\sc RideOnCycle}, which computes an optimal
ride for any ride-sharing scenario $\mathcal{R}=\tuple{G,(s_0,t_0),\mathcal{C}}$, with $G$ being a cycle.
The algorithm founds on the idea of enumerating each possible tuple $\tuple{\alpha,\beta,v_{s_0},v_{t_0},s^\circ,t^\circ}$ of elements as in
Theorem~\ref{thm:cycle}. For each given configuration, the optimal ride $\pi^\circ$ over the scenario
$\tuple{G^\circ,(v_{s_0},v_{t_0}),\mathcal{C}_{\alpha,\beta}^\circ\cup\{(s^\circ,t^\circ)\}}$ is computed. Eventually, $\pi^*$ is defined (see
step~\ref{step:confronto}) as the ride with minimum cost (w.r.t. $w^\circ$) over such rides $\pi^\circ$. The ride $\pi^*_1\ {\tt mod}\ n,
\dots,\pi^*_{\len(\pi^*)}\ {\tt mod}\ n$ is then returned.

\begin{algorithm}[t]
\SetKwInput{KwData}{Input} \SetKwInput{KwResult}{Output}

\Indm
\KwData{A ride-sharing scenario $\mathcal{R}=\tuple{G,(s_0,t_0),\mathcal{C}}$, where $G$ is a cycle;}%
\KwResult{An optimal ride for $\mathcal{R}$ ;}

\Indp

   \For{each tuple $\tuple{\alpha,\beta,v_{s_0},v_{t_0},s^\circ,t^\circ}$ of elements as in Theorem~\ref{thm:cycle}}{

   Let $\pi^\circ$ be an optimal ride for $\tuple{G^\circ,(v_{s_0},v_{t_0}),\mathcal{C}_{\alpha,\beta}^\circ\cup\{(s^\circ,t^\circ)\}}$\;

	   \If{$\pi^*$ is not yet defined or $w^\circ(\pi^\circ) < w^\circ(\pi^*)$}{\label{step:confronto}
			$\pi^* \GETS \pi^\circ$\;
		}

}  \Return{$\pi^*_1\ {\tt mod}\ n, \dots,\pi^*_{\len(\pi^*)}\ {\tt mod}\ n$;}

\caption{{\sc RideOnCycle}}\label{alg:cycle}
\end{algorithm}

\begin{theorem}
Algorithm {\sc RideOnCycle} is correct.
\end{theorem}
\begin{proof}
In order to analyze the correctness, observe that by Theorem~\ref{thm:cycle}, the ride returned as output, say $\Lambda^*=\pi^*_1\ {\tt mod}\
n, \dots,\pi^*_{\len(\pi^*)}\ {\tt mod}\ n$, is necessarily feasible for $\R$. Therefore, assume for the sake of contradiction that there is an
optimal ride $\pi$ for $\R$ such that $w(\pi)<w(\Lambda^*)$. In particular, by construction of $w^\circ$, we derive that
$w(\pi)<w(\Lambda^*)=w^\circ(\pi^*)$.

Now, by Lemma~\ref{lem:boundCW}, we can actually assume, w.l.o.g., that $\cw(\pi)\leq 3n$ and $\{\cw(\pi)\ {\tt mod}\ n,\acw(\pi)\ {\tt
mod}\ n\}\subseteq V_\mathcal{C} \cup \{s_0,t_0\}$ hold. So, we can apply Theorem~\ref{thm:cp} and derive the existence of a tuple
$\tuple{\alpha,\beta,v_{s_0},v_{t_0},s^\circ,t^\circ}$ of elements, with $v_{s_0}=\tau_\pi(1)$ and $v_{t_0}=\tau_\pi(\len(\pi))$, satisfying
properties \emph{(i)}, \emph{(ii)}, and \emph{(iii)} in Theorem~\ref{thm:cycle} and such that $\Upsilon=\tau_\pi(1),...,\tau_\pi(\len(\pi))$ is
feasible for $\tuple{G^\circ,(v_{s_0},v_{t_0}),\mathcal{C}_{\alpha,\beta}^\circ\cup\{(s^\circ,t^\circ)\}}$. In particular, by construction of
$w^\circ$, we derive that $w^\circ(\Upsilon)=w(\pi)$.
However, the algorithm has compared the weight of $\Upsilon$ and $\pi^*$, and hence we know that $w(\pi)=w^\circ(\Upsilon)\geq w^\circ(\pi^*)$,
which is impossible.
\end{proof}

\smallskip

Let us finally discuss about the implementation and running time of  the algorithm. Before starting the loop, we first compute the sets $W = \{w \in V^\circ \mid
1\leq w\leq 3n \mbox{ and } (w  \mod n) \in V_\C \cup \{s_0,t_0\}\}$ and $\C^\circ = \{(s,t) \in W \mid (s\mod n, t\mod n) \in \C\}$; this can
be done in time $O(|\C|)$ by iterating through the requests in $\C$. Note that $|W| = O(|V_\C|)$ and $\C^\circ| = O(|\C|)$.
Now, note that the number of iterations of {\sc RideOnCycle} corresponds to the number tuples
$\tuple{\alpha,\beta,v_{s_0},v_{t_0},s^\circ,t^\circ}$ which satisfy the conditions of Theorem~\ref{thm:cycle}.
The number of possible pairs $(\alpha,\beta)$ is $W^2 = O(|V_\C|^2)$.
Checking whether condition \emph{(i)} in Theorem~\ref{thm:cycle} holds on them can be simply accomplished by checking that every element $x\in
V_\C \cup \{s_0,t_0\}\}$ is such that $\alpha\mod n\leq x \leq \beta\mod n$. So, it can be done in constant time after that, in a
pre-processing step costing $O(|V_C|)$, the minimum and maximum element in $V_\C \cup \{s_0,t_0\}\}$ have been computed.
Moreover, note that since $1\leq \alpha,\beta\leq 3n$, according to Theorem~\ref{thm:cycle}, there are at most 3 possible choices for $s_0$
(resp, $t_0$); in addition, there are just two alternatives for the pair $s^\circ,t^\circ$. Hence, summarizing we have that all tuples
satisfying the conditions of Theorem~\ref{thm:cycle} can be actually build in $O(|V_\C|^2)$.
Then, by inspecting the operations performed at each iteration, for each tuple $\tuple{\alpha,\beta,v_{s_0},v_{t_0},s^\circ,t^\circ}$, we have
to compute the set $\mathcal{C}_{\alpha,\beta}^\circ$. To this end, we search among the elements in $\C^\circ$ for the pairs $(s,t)$ having
both nodes in $V^\circ_{\alpha,\beta}$; this step takes $O(|\C|)$. Finally, on the resulting scenario defined on a path, we apply the algorithm
for computing an optimal ride, which costs $O(|\C| \log |\C|+|V|)$. 
Hence the following theorem follows. 

\begin{theorem}\label{thm:impl:cycle}
Let $\mathcal{R}=\tuple{G,(s_0,t_0),\mathcal{C}}$ be a ride-sharing scenario where $G=(V,E,w)$ is a cycle. Then, an optimal ride for $\R$
(together with its cost) can be computed in time $O(|V_\C|^2\cdot (|\C| \log |\C|+|V|))$.
\end{theorem}


%


\section{Related Work}\label{sec:related}
\paragraph{Ride Sharing.}
Based on whether or not we allow objects to be temporarily unloaded at some vertex of the transportation network, two versions of ride sharing
problems emerge: \emph{preemptive} (where drops are allowed) and \emph{non-preemptive} (where drops are not allowed).
An orthogonal classification comes, moreover, from the capacity $c$ of the given vehicle. The setting with \emph{unit} capacity ($c=1$) has
received much attention in the literature, where it often comes in the form of a \emph{stacker crane problem}
(see~\cite{doi:10.1137/0207017,JordanSrour2013674} and the references therein). A natural generalization is then when the vehicle can carry
more than one object at time, that is, when $c$ is any given natural number possibly larger than 1.

\begin{figure}[h]\centering \small
\begin{tabular}{cc}
  \begin{tabular}{c}
  \begin{tabular}{|c||c|c|}
    \hline
                    & \emph{preemptive} & \emph{non-preemptive} \\
    \hline \hline
    \textbf{trees} & in $\bf P$~\cite{doi:10.1137/0221066} & $\bf NP$-hard~\cite{Frederickson199329}\\
    \textbf{cycles} & in $\bf P$~\cite{Frederickson:1993:NCS:152322.152334} & in $\bf P$~\cite{Frederickson:1993:NCS:152322.152334}\\
    \textbf{paths}  & in $\bf P$~\cite{doi:10.1137/0217053} & in $\bf P$~\cite{doi:10.1137/0217053}\\
    \hline
  \end{tabular}\\
  $c=1$
  \end{tabular}
       &
    \begin{tabular}{c}
  \begin{tabular}{|c||c|c|}
    \hline
                    & \emph{preemptive} & \emph{non-preemptive} \\
    \hline \hline
    \textbf{trees} & $\bf NP$-hard~\cite{Guan199841} & $\bf NP$-hard~\cite{Frederickson199329}\\
    \textbf{cycles} & in $\bf P$~\cite{guan1998multiple}$^*$ & $\bf NP$-hard~\cite{Guan199841}\\
    \textbf{paths}  & in $\bf P$~\cite{guan1998multiple} & $\bf NP$-hard~\cite{Guan199841}\\
    \hline
  \end{tabular}\\
  $c\geq 1$
  \end{tabular}
  \end{tabular}
   \caption{Summary of results in the literature. $^*$It is assumed that, for each object, the direction of its transportation (either clockwise, or anticlockwise) is is a-priori fixed.}\label{tab:summary}
\end{figure}

Given these two orthogonal dimensions, a total of four different configurations can be studied (cf.~\cite{guan1998multiple}).
In all the possible configurations, vehicle routing is known to be $\bf NP$-hard~\cite{doi:10.1137/0207017,Garey:1979:CIG:578533} when the
underlying transportation network is an arbitrary graph.
In fact, motivated by applications in a wide range of real-world scenarios, complexity and algorithms for ride sharing problems have been
studied for networks with specific topologies, such as path, cycles, and trees.
A summary of the results in the literature referring to these studies is reported in Figure~\ref{tab:summary}.
By looking at the table, consider first the unit capacity setting. In this case, ride sharing is known to be polynomial time solvable on
both paths~\cite{doi:10.1137/0217053} and cycles~\cite{Frederickson:1993:NCS:152322.152334}, no matter of whether drops are allowed. Moving to
trees, instead, the preemptive case remains efficiently solvable~\cite{doi:10.1137/0221066}, while the non-preemptive case becomes $\bf
NP$-hard~\cite{Frederickson199329}.

Consider now the case where $c\geq 1$ holds. Clearly enough, the intractability result over trees established for $c=1$ still holds in this
more general setting.
In fact, in this setting, ride sharing appears to be intrinsically more complex. Indeed, it has been shown that the non-preemptive version of
the problem is $\bf NP$-hard on all the considered network topologies and that the preemptive version is $\bf NP$-hard even on
trees~\cite{Guan199841}.
Good news comes instead when the problem is restricted over paths and cycles in the preemptive case.
Indeed, the problem has been shown to be feasible in polynomial time on paths, formally in $O((k+n)\times n)$ where $k$ is the number of
objects and $n$ is the number of vertices~\cite{guan1998multiple}.
Moreover, the algorithm proposed by~\cite{guan1998multiple} is also applicable to cycles, under the constraint that, for each object, the
direction of the transportation (either clockwise, or anticlockwise) is a-priori given.
More efficient algorithms are know for paths in the special case where the ride starts from one endpoint~\cite{Guan199841,Karp}.

\paragraph{Vehicles of Unlimited Capacity.} The $\bf NP$-hardness results discussed above exploit a given constant bound on the capacity and, hence, they do not immediately
apply to the unbounded setting. However, specific reductions have been exhibited showing the $\bf NP$-hardness on general graphs
(cf.~\cite{Tzoreff:2002:VRP:508259.508261,Chalasani:1999:ACR:323331.323342}). Moreover, heuristic methods (see,
e.g.,~\cite{Gendreau1999699,Mosheiov1998669}) and approximation algorithms (see,
e.g.,~\cite{Asano:1997:CPP:258533.258602,doi:10.1287/moor.10.4.527}) have been defined, too.
On the other hand, a number of tractability results for vehicles with unlimited capacity transporting objects of the same type can be inherited
even in the paired context we are considering. Indeed, by focusing on problems where such identical objects are initially stored at the same
node (or, equivalently, have to be transported to the same
destination)~\cite{Chalasani:1999:ACR:323331.323342,chalasani1996algorithms,charikar2001algorithms,743496}, efficient algorithms have been
designed for transportation networks that are trees and cycles~\cite{Tzoreff:2002:VRP:508259.508261}, with the running time being $O(n)$ and
$O(n^2)$, respectively.
Moreover, the algorithm for paths (and cycles, with the limitation discussed above) proposed by~\cite{guan1998multiple} can be still applied
over the unlimited capacity scenario. However, it was not explored in the literature whether its performances can be improved by means of
algorithms specifically designed for vehicles with unlimited capacity. Addressing this open issue is the distinguishing feature of the research
reported in the paper. Moreover, differently from~\cite{guan1998multiple}, our algorithm to solve the ride sharing problem over cycles does not
require that the direction of the transportation of the objects is fixed beforehand.
%

\section{Conclusion}\label{sec:conclusion}
We have consider a ride sharing problem with a vehicle of unlimited capacity, by completely classifying its complexity w.r.t.~the underlying
network topology. The main result is a $O(|C| \log |C|)$ algorithm for computing an optimal ride over paths, with $C$ denoting the set of the
available requests.
Our results have a wide spectrum of applicability, in particular, to find optimal rides whenever it is a-priori known that the number of
objects to be transported does not exceed the capacity of the vehicle.

In fact, computing an optimal ride might be not enough in some applications. Indeed, especially in the context of transportation of passengers
(such as in \emph{dial-a-ride} problems~\cite{DRP}), the human perspective tend to introduce further requirements leading to balance user
inconvenience against minimizing routing costs; in particular, the time comparison of the chosen route with respect to the shortest path to a
destination is a widely-used measure of customer satisfaction in (the related) \emph{school bus routing problems}~\cite{Park2010311}.
Accordingly, an interesting avenue for further research is to adapt our solution algorithms by taking into account fairness requirements.
Finally, we stress here that another interesting technical question is to assess whether, in our basic optimization setting, further
tractability results can be established by focusing on requests of special kinds, for instance, on requests where the starting and terminating
nodes precisely identify the endpoints of some edge. In this latter case, it would be interesting to analyze the complexity over trees (which
emerged to be intractable with arbitrary requests) and, more generally, over graphs having bounded treewidth.

\paragraph{Acknowledgment.}
This work was partially supported by the project ANR-14-CE24-0007-01 \emph{``CoCoRICo-CoDec"}.
We thank J\'{e}r\^{o}me Lang, from Universit\'{e} Paris-Dauphine, for introducing the subject to us.



\end{document}